\documentclass[english, a4paper]{article}

\usepackage[margin=1.65in]{geometry}
\usepackage[font=small, labelfont=bf, margin=0.35in]{caption}

\usepackage{titling}
\usepackage{authblk}

\makeatletter
\newenvironment{breakablealgorithm}
  {
   \begin{center}
     \refstepcounter{algorithm}
     \hrule height.8pt depth0pt \kern2pt
     \renewcommand{\caption}[2][\relax]{
       {\raggedright\textbf{\fname@algorithm~\thealgorithm} ##2\par}%
       \ifx\relax##1\relax
         \addcontentsline{loa}{algorithm}{\protect\numberline{\thealgorithm}##2}%
       \else
         \addcontentsline{loa}{algorithm}{\protect\numberline{\thealgorithm}##1}%
       \fi
       \kern2pt\hrule\kern2pt
     }
  }{
     \kern2pt\hrule\relax
   \end{center}
  }
\makeatother

\usepackage{babel}
\usepackage[utf8]{inputenc}
\usepackage[pdftex]{hyperref}

\usepackage[mathscr]{euscript}

\usepackage{color}
\usepackage{graphicx}
\usepackage{relsize}
\usepackage{enumitem}
\usepackage{pgf}

\usepackage{tikz}
\usetikzlibrary{arrows}

\usepackage{amsmath}
\usepackage{amssymb}
\usepackage{amsthm}

\usepackage{thmtools}
\usepackage{thm-restate}

\usepackage{multirow}

\usepackage[labelfont=bf, figurename=Fig., tablename=Tab.]{caption}

\allowdisplaybreaks

\newcommand{\refeq}[1]{(\ref{#1})}

\newcommand{\ket}[1]{\left|\, #1 \, \right\rangle}
\newcommand{\e}{\mathrm{e}}
\newcommand{\dd}{\mathrm{d}}
\newcommand{\imag}{i}
\newcommand{\ordo}{O} 
\newcommand{\ceil}[1]{\left\lceil #1 \right\rceil}
\newcommand{\floor}[1]{\left\lfloor #1 \right\rfloor}
\newcommand{\round}[1]{\left\lfloor #1 \right\rceil}
\newcommand{\sgn}{\text{sgn}}
\newcommand{\lcm}{\text{lcm}}
\newcommand{\poly}{\mathrm{poly}}
\renewcommand{\vec}{\mathbf}
\newcommand{\Si}{\mathop{\text{Si}}}
\renewcommand{\parallel}{=}
\newcommand{\inset}{\cap}

\newcommand{\deltaexpr}{6 \sqrt{3} \cdot 2^{\Delta}}

\newtheorem{theorem}{Theorem}
\newtheorem{lemma}{Lemma}

\newtheorem{theoremcorollary}{Corollary}[theorem]

\newtheorem{appdefinition}{Definition}[section]
\newtheorem{appclaim}{Claim}[section]
\newtheorem{applemma}{Lemma}[section]

\newenvironment{notes}{\noindent\emph{Notes.}}{\ignorespacesafterend}

\usepackage{enumitem}
\usepackage{algorithm}

\newlist{pseudocode}{enumerate}{4}
\setlist[pseudocode]{
  label={\arabic*},
  ref={\arabic*},
  before=\raggedright,
  leftmargin=*,
  itemsep=0pt,
  topsep=2pt}
\setlist[pseudocode,2]{label*={.\arabic*}, topsep=1pt}
\setlist[pseudocode,3]{label*={.\arabic*}, topsep=1pt}
\setlist[pseudocode,4]{label*={.\arabic*}, topsep=1pt}

\newlist{romanlist}{enumerate}{1}
\setlist[romanlist]{
  label={\roman*)},
  leftmargin=*,
  itemsep=3pt,
  topsep=3pt}

\title{On the success probability of quantum order finding}
\author[1,2]{\href{mailto:ekera@kth.se}{Martin Ekerå}}
\affil[1]{\small KTH Royal Institute of Technology, Stockholm, Sweden}
\affil[2]{\small Swedish NCSA, Swedish Armed Forces, Stockholm, Sweden}

\predate{\begin{center}}
\postdate{\end{center}\vspace{-3mm}}

\begin{document}
\maketitle

\begin{abstract}
We prove a lower bound on the probability of Shor's order-finding algorithm successfully recovering the order~$r$ in a single run.
The bound implies that by performing two limited searches in the classical post-processing part of the algorithm, a high success probability can be guaranteed, for any~$r$, without re-running the quantum part or increasing the exponent length compared to Shor.
Asymptotically, in the limit as~$r$ tends to infinity, the probability of successfully recovering~$r$ in a single run tends to one.
Already for moderate~$r$, a high success probability exceeding e.g.~$1 - 10^{-4}$ can be guaranteed.
As corollaries, we prove analogous results for the probability of completely factoring any integer~$N$ in a single run of the order-finding algorithm.
\end{abstract}

\section{Introduction}
In~\cite{shor94, shor97}, Shor famously introduced a polynomial-time quantum algorithm for factoring integers.
Shor's factoring algorithm works by first classically reducing the integer factoring problem (IFP) to an order-finding problem (OFP).
The resulting OFP is then solved using an order-finding algorithm, that in turn consists of a quantum part and a classical post-processing part.

Although the order-finding algorithm was originally described by Shor in the context of integer factorization, the algorithm is generic:
Let~$g$ be a generator of a finite cyclic group, that we write multiplicatively in what follows, and for which the group operation may be implemented efficiently quantumly.
Shor's order-finding algorithm then finds the order~$r$ of~$g$ in polynomial time.
It is also possible to view the algorithm as a generic period-finding algorithm~\cite[p.~1501]{shor97}.

In this work, we provide a detailed analysis of the success probability of Shor's order-finding algorithm.
In particular, we prove a lower bound on the probability of the algorithm successfully recovering~$r$ in a single run, and we show this probability to be high already for moderate~$r$.
Asymptotically, in the limit as~$r$ tends to infinity, the success probability tends to one.
A key to achieving these results is to perform two limited searches in the classical post-processing.
We review earlier results and analyses in Sect.~\ref{sec:earlier-works} for comparison to the results we obtain in this work.

As a corollary, we use the reduction from the IFP to the OFP in~\cite{completely} to prove a lower bound on the probability of completely factoring any integer~$N$ efficiently in a single run of the order-finding algorithm.
Compared to the bound in~\cite{completely}, our lower bound also accounts for the probability of the order-finding algorithm failing to recover~$r$.
Again, already for moderate~$N$, a high probability of recovering the complete factorization of~$N$ in a single run can be guaranteed.
Asymptotically, in the limit as~$N$ tends to infinity, the success probability tends to one.

\subsection{Preliminaries}
Assuming quantum computation to be more expensive than classical computation, it is advantageous to strive to reduce the requirements imposed on the quantum computer at the expense of increasing the classical computational burden.

This provided that it remains feasible to execute the classical part in practice.
A minimum requirement is that it must execute in polynomial time.
Against this backdrop, we parameterize our algorithms:
We pick parameters in our asymptotic analyses so that both the quantum and classical parts execute in polynomial time.
For concrete problem instances, we pick parameters so that the classical part can be executed in practice on, for instance, an ordinary workstation or cluster.

In Shor's order-finding algorithm, a generator~$g$ is exponentiated to an exponent in superposition.
In a typical implementation, the standard square-and-multiply algorithm is used to classically pre-compute powers~$g^{2^i}$ of~$g$ that are then composed quantumly.
Hence, as a rule of thumb, the longer the exponent, the longer the runtime and the required coherence time, and the larger the circuit depth, and so forth.
Consequently, we strive to keep the exponent as short as possible, whilst respecting the constraint that a single run shall in general suffice to compute~$r$.

Compared to Shor's original work on performing order-finding in the context of integer factorization, the exponent length is on par or slightly shorter in this work\footnote{Depending on how the order $r$ is upper-bounded, for further details see App.~\ref{appendix:notes-shor-order-finding-factoring}.} when using traditional continued fractions-based classical post-processing.
A few more bits can be peeled off by instead using lattice-based post-processing.

It is possible to further reduce the exponent length, at the expense of performing several runs of the quantum part of the algorithm and jointly post-processing the outputs classically, as in Seifert's variation~\cite{seifert} of Shor's order-finding algorithm.

Although several parts of our analysis cover both Shor's original order-finding algorithm and Seifert's variation of it, the focus in this work is on analyzing the success probability when performing order finding with the aim of computing~$r$ in a single run of the quantum part of the algorithm.
For a detailed analysis of Seifert's algorithm, the interested reader is instead referred to~\cite[App.~A]{general}.

Note also that if the reason for performing order finding is to factor RSA integers so as to e.g.\ break the RSA cryptosystem~\cite{rsa}, then the specialized quantum algorithm of Ekerå and Håstad~\cite{ekera-pp, ekera-hastad} outperforms Shor's and Seifert's algorithms.

\subsection{Quantum order finding}
\label{sec:quantum-order-finding}
In what follows, we re-use notation and elements from~\cite[App.~A]{general}:
We let~$g$ be a generator of a cyclic group of order $r \ge 2$, and induce the quantum state
\begin{align}
  \label{eq:quantum-state}
  \frac{1}{2^{m+\ell}} \sum_{a, \, j \, = \, 0}^{2^{m+\ell - 1}}
  \exp \left[ \frac{2\pi\imag}{2^{m+\ell}} aj \right]
  \ket{j, g^a}
\end{align}
for~$m$ an upper bound on the bit length of~$r$ such that $r < 2^m$, and $\ell \sim m/s$ a positive integer for~$s$ some tradeoff factor.
In Shor's original algorithm~\cite{shor94, shor97}, the tradeoff factor $s = 1$, whereas $s > 1$ in Seifert's variation~\cite{seifert} of the algorithm.
As explained in~\cite[App.~A]{general}, a single run yields at least\footnote{If~$m$ is the bit length of~$r$, the algorithm yields $\sim \ell$ bits of information on~$r$.
If~$m$ is greater than the bit length of~$r$, the algorithm yields more than $\sim \ell$ bits of information on~$r$.} $\sim \ell$ bits of information on~$r$.

Let $\mathcal A(e) = \{ t \in [0, 2^{m+\ell}) \inset \mathbb Z \mid t \equiv e \:\: (\text{mod } r) \}$.
Following~\cite[App.~A]{general}, if we measure the first register of~\refeq{eq:quantum-state}, then we observe the frequency~$j$ with probability
\begin{align}
  &\frac{1}{2^{2(m+\ell)}}
  \sum_{e \, = \, 0}^{r-1} \,
  \left|\,
  \sum_{a \, \in \, \mathcal A(e)}
  \exp \left[ \frac{2\pi\imag}{2^{m+\ell}} aj \right]
  \,\right|^2 = \notag \\
  &\frac{1}{2^{2(m+\ell)}}
  \sum_{e \, = \, 0}^{r-1} \,
  \left|\,
  \sum_{b \, = \, 0}^{\floor{(2^{m+\ell} - e - 1) / r}}
  \exp \left[ \frac{2\pi\imag}{2^{m+\ell}} (e + rb)j \right]
  \,\right|^2 = \notag \\
  &\frac{1}{2^{2(m+\ell)}}
  \sum_{e \, = \, 0}^{r-1} \,
  \left|\,
  \sum_{b \, = \, 0}^{\floor{(2^{m+\ell} - e - 1) / r}}
  \e^{\imag \theta_r b}
  \,\right|^2 = \notag \\
  &\frac{\beta}{2^{2(m+\ell)}} \left|\, \sum_{b \, = \, 0}^{L} \e^{\imag \theta_r b} \,\right|^2
  +
  \frac{r - \beta}{2^{2(m+\ell)}} \left|\, \sum_{b \, = \, 0}^{L - 1} \e^{\imag \theta_r b} \,\right|^2 \label{eq:pr}
\end{align}
where~$\beta = 2^{m+\ell} \text{ mod } r$, $L = \floor{2^{m+\ell} / r}$, and
\begin{align*}
  \theta_r = \frac{2 \pi \alpha_r}{2^{m+\ell}}
  \quad \text{ for } \quad
  \alpha_r = \{ rj \}_{2^{m+\ell}}
\end{align*}
where $\{u\}_N$ denotes $u \text{ mod } N$ constrained to $[-N/2, N/2)$.
In what follows, we refer to~$\theta_r$ and~$\alpha_r$ as angles and arguments, respectively, associated with~$j$.

Furthermore, $|\, u \,|$ denotes the absolute value of $u \in \mathbb C$, whereas $\floor{f}$, $\ceil{f}$ and $\round{f}$ denote $f \in \mathbb R$ rounded down, up or to the closest integer, respectively.
Ties are broken by requiring that $\round{f} = f + \delta_f$ for some $\delta_f \in (-1/2, 1/2]$.

If $\alpha_r \neq 0$, we may write~\refeq{eq:pr} on closed form as
\begin{align}
  \label{eq:P-non-zero}
  P(\alpha_r)
  &=
  \frac{\beta}{2^{2(m+\ell)}} \frac{1 - \cos(2 \pi \alpha_r (L + 1) / 2^{m+\ell})}{1 - \cos(2 \pi \alpha_r / 2^{m+\ell})}
  + \hphantom{x} \notag \\
  &\quad\quad\quad
  \frac{r - \beta}{2^{2(m+\ell)}} \frac{1 - \cos(2 \pi \alpha_r L / 2^{m+\ell})}{1 - \cos(2 \pi \alpha_r / 2^{m+\ell})}.
\end{align}

Otherwise, if $\alpha_r = 0$, we may write~\refeq{eq:pr} on closed form as
\begin{align}
  \label{eq:P-zero}
  P(0)
  =
  \frac{\beta}{2^{2(m+\ell)}} (L+1)^2
  +
  \frac{r - \beta}{2^{2(m+\ell)}} L^2
  =
  \frac{L^2 r + (2L + 1)\beta}{2^{2(m+\ell)}}.
\end{align}

In what follows, we furthermore let $\ln u$, $\log u$ and $\log_q u$ denote the natural, base-two and base-$q$ logarithms of $u$.
We let~$\mathbb Z_{>0}$ and~$\mathbb Z_{\ge 0}$ denote the positive and non-negative integers, respectively.

\subsection{Understanding the probability distribution}
\label{sec:understanding-the-probability-distribution}
The probability distribution induced by the order-finding algorithm has~$r$ peaks, located approximately a distance $2^{m+\ell} / r$ apart in~$j$.
The first peak is at $j = 0$.

More specifically, the optimal integer frequency of the peak with index~$z$ is
\begin{align*}
  j_0(z) = \round{\frac{2^{m+\ell}}{r} z}
  \quad \text{ for } \quad
  z \in [0, r) \inset \mathbb Z,
\end{align*}
which implies, for $\alpha_0(z)$ the argument associated with $j_0(z)$, that
\begin{align*}
  \alpha_0(z)
  &=
  \{ r j_0(z) \}_{2^{m+\ell}}
  =
  \{ r \round{2^{m+\ell} z / r} \}_{2^{m+\ell}}
  =
  \{ r (2^{m+\ell} z / r + \delta_z) \}_{2^{m+\ell}} \\
  &=
  \{ 2^{m+\ell} z + \delta_z r \}_{2^{m+\ell}}
  =
  \{ \delta_z r \}_{2^{m+\ell}} = \delta_z r \in (-r/2, r/2]
\end{align*}
for some $\delta_z \in (-\frac{1}{2}, \frac{1}{2}]$ that depends on~$z$.

If~$r$ is a power of two, we get~$r$ distinct peaks for which $P(\alpha_0(z)) = 1/r$, whereas $P(\alpha_r) = 0$ for all other~$\alpha_r$.
If~$r$ is not a power of two, $P(\alpha_0(z))$ will vary slightly between the~$r$ peaks, and there will be some noise around the peaks, due to rounding effects.
Hence, the probability mass will be slightly spread out.

The frequency $j_0(z)$ is optimal for the peak with index~$z$ in the sense that other frequencies in the neighborhood of $j_0(z)$ yield smaller or at most equal contributions to the probability mass associated with the peak.

\subsection{Shor's original analysis and post-processing}
\label{sec:shors-original-work}
In his original paper, Shor~\cite{shor94, shor97} lower-bounds the success probability of quantum order finding:
Specifically, Shor shows~\cite[p.~1500]{shor97} that the probability of observing $j_0(z)$ for $z \in [0, r) \inset \mathbb Z$ is asymptotically lower-bounded by $4/(r \pi^2)$ independent of~$z$.
Hence, the probability of observing $j_0(z)$ for some $z \in [0, r) \inset \mathbb Z$ may asymptotically be taken as $4/\pi^2$, where, critically, $z$ is then uniformly distributed on $[0, r) \inset \mathbb Z$.

For $j_0(z)$, the convergent $z/r$ appears in the continued fraction expansion of $j_0(z) / 2^{m + \ell}$ if~$\ell$ is selected sufficiently large (so that $2^{m+\ell} > r^2$, see~\refeq{eq:hw-requirement-j0} in Sect.~\ref{sec:find-order}).

Given the convergent $z/r$, the order~$r$ may of course be trivially recovered, unless cancellations occur due to~$z$ and~$r$ not being coprime.
Shor points out~\cite[p.~1501]{shor97} that at least a fraction $\phi(r)/r$ of the~$r$ values of $z \in [0, r) \inset \mathbb Z$ are coprime to~$r$, and that it follows from~\cite[Thm.~328 on p.~267]{hw} that $\phi(r)/r = \Omega(1 / \log \log r)$.

In summary, the above analysis gives a lower bound on the success probability, and hence an upper bound on the expected number of runs required to recover~$r$.
This being said, the procedure may be improved, as Shor states~\cite[p.~1501]{shor97}:

Firstly, even if an optimal~$j$ is not observed, the~$j$ observed may be close to~$j_0(z)$.
Hence, by trying to solve not only~$j$, but also $j \pm 1$, $j \pm 2$‚ $\ldots$, for $z/r$, the success probability may be increased beyond the asymptotic lower bound of $4/\pi^2$.

Secondly, Shor credits Odlyzko~\cite[p.~1501]{shor97} for pointing out in personal communication that for $d = \gcd(r, z)$, one may recover $r/d$ from the denominator of $z/r$ and then search over~$d$ to recover~$r$ from $r/d$.
This improves the expected number of runs from $\ordo(\log \log r)$ without searching to $\ordo(1)$ provided one exhausts on the order of $(\log r)^{1+\epsilon}$ values of~$d$.\footnote{The statement in~\cite{shor97} is made in the context of factoring an integer $N$, and is expressed in terms of~$N$, by using that $r < N$ for $r$ the order of $g \in \mathbb Z_N^*$. We restate the claim in terms of~$r$.}

\subsection{Earlier works that follow up on Shor's work}
\label{sec:earlier-works}
Following the publication of Shor's original work, a number of analyses of the success probability of quantum order finding have been entered into the literature:

For instance, Knill~\cite{knill} points out that the probability of recovering the order~$r$ may be increased by running Shor's order-finding algorithm multiple times for the same~$g$ and taking the least common multiple of the tentative orders $r / \gcd(r, z)$ extracted from the convergents $z/r$.
After two runs, that yield convergents $z_1 / r$ and $z_2 / r$, respectively, with~$z_1$ and~$z_2$ selected uniformly at random from $[0, r) \inset \mathbb Z$, the probability of finding~$r$ from $z_1/r$ and $z_2/r$ is then lower-bounded by a constant.

Knill furthermore explores tradeoffs between the exponent length, the search space in the classical post-processing, and the probability of recovering $z/r$ from the frequency~$j$ in the classical post-processing.
Cleve et al.~\cite[Sect.~5 on p.~348, App.~C]{cemm97} also explore such tradeoffs.
Seifert~\cite{{seifert}}, and later Ekerå~\cite[App.~A]{general}, explore tradeoffs in the context of jointly solving a set of frequencies for~$r$ --- the idea being to perform less work in each run of the quantum part of the algorithm, at the expense of performing more runs of the quantum part, and more work overall.

McAnally~\cite{mcanally} increases the probability of finding some non-trivial divisor of the order in a single run to ``negligibly less than one'', at the expense of increasing the exponent length by approximately a factor of~$1.5$ compared to Shor.

Asymptotically, the probability of finding the order~$r$ in~$n$ runs is reported to be greater than $1/\zeta(n) - \ordo(N^{-\epsilon})$, when performing order finding in~$\mathbb Z_N^*$, and for~$\zeta$ the Riemann zeta function, as~$N$ tends to infinity for~$\epsilon$ a positive number.
The probability is hence at least~$0\%$ after one run --- i.e.\ we get no information, slightly greater than~$60\%$ after two runs, and slightly greater than~$90\%$ after four runs.

Gerjuoy~\cite[Sect.~III.C.8 on p.~534]{gerjuoy}, provides a lower bound of~$90\%$ on the probability of recovering some non-trivial divisor of the order in a single run, when using the same exponent length~\cite[Sect.~III.C.1 on p.~529]{gerjuoy} as Shor.
This when performing order finding in~$\mathbb Z_N^*$, for~$N$ an odd integer with at least two distinct prime factors as in Shor's original work.
Gerjuoy accomplishes this improvement by leveraging that $r \le \lambda(N) < N/2$ in his analysis~\cite[Sect.~III.C.7 on p.~533]{gerjuoy}, for~$\lambda$ the Carmichael function, whereas Shor uses~$N$ as an upper bound on~$r$.

Bourdon and Williams~\cite{bourdon} report sharp bounds on the probability of Shor's order-finding algorithm returning some non-trivial divisor of the order in a single run; both when the exponent length is as in Shor's original algorithm, and when the exponent is allowed to increase in length by~$q$ bits compared to Shor.

In the former case, they improve upon Gerjuoy's result and obtain an asymptotic lower bound on the probability of $2 \Si(4\pi) / \pi \approx 0.9499$, where~$\Si$ is the sine integral function.
In the latter case, they obtain an asymptotic lower bound of $2 \Si(2^{q+2} \pi) / \pi$.
Both bounds pertain to order finding in~$\mathbb Z_N^*$, for~$N$ an odd integer with at least two distinct prime factors, as in Shor's and Gerjuoy's works.

Proos and Zalka~\cite[App.~A.1]{proos-zalka} briefly discuss how the success probability of Shor's order-finding algorithm is affected when attempting to solve not only the frequency observed~$j$, but also $j \pm 1$, $\ldots$, $j \pm B$, for the order~$r$.

Einarsson~\cite{einarsson} investigates the probability distribution induced by Shor's order-finding algorithm, and the expected number of runs.
He gives no formal lower bound on the success probability, but concludes that there is a high probability of the algorithm yielding a non-trivial divisor of~$r$.
This when performing order finding in~$\mathbb Z_N^*$, for~$N$ a semiprime, with an exponent of the same length as Shor.

Furthermore, Einarsson points out that one can test if a candidate for~$r$ is a positive integer multiple of $r$.
If only a non-trivial divisor of~$r$ is found, then the missing factor may be found ``by trial or possibly by a more efficient algorithm''.

Davis~\cite{davis} follows up on the work by Bourdon and Williams with a slightly improved analysis in the context of benchmarking quantum computers.

In the context of general quantum phase estimation, where a phase is to be estimated with~$s$ bits of precision using what the authors' refer to as~$p$ extra qubits (i.e.\ extra exponent bits), Chappell et al.~\cite{clsia} present an exact formula for the error probability $\epsilon(s, p)$, and explore the limits of this formula as $s \rightarrow \infty$ and $p \rightarrow \infty$.

\subsubsection{Precise simulations and estimates}
\label{sec:earlier-works-simulations}
In~\cite[App.~A]{general}, Ekerå studies the success probability of quantum order finding by means of simulations, when using lattice-based post-processing, and with respect both to Shor's original order-finding algorithm and to Seifert's algorithm.

The idea is to first capture the probability distribution induced for a given known order~$r$, and parameters~$m$ and~$\ell$, and to then sample the probability distri\-bution to simulate the quantum part of the algorithm.
This is feasible for large, crypto\-graphically relevant, problem instances for as long as~$r$ is known, enabling accurate estimates of the success probability and the expected number of runs to be derived for various post-processing strategies and parameterizations.

The simulations show~\cite[App.~A]{general} that a single run of Shor's order-finding algorithm with $\ell \sim m$ usually suffices to recover~$r$, provided a limited enumeration of the lattice is performed.\footnote{Note that~$m$ is the bit length of~$r$ in~\cite{general}, and an upper bound on the bit length in this paper.}
For Seifert's algorithm with $s \ge 1$ an integer, it suffices to perform slightly more than~$s$ runs to achieve $\ge 99\%$ success probability, see e.g.~\cite[Tab.~A1--A2]{general}, when not enumerating the lattice. This may be slightly improved by enumerating a limited number of vectors in the lattice.

As a rule of thumb, it is possible to decrease~$\ell$ at the expense of enumerating more vectors in the lattice, and vice versa.
The enumeration of the lattice essentially captures the two searches in offsets in~$j$ and in multiples of the denominator in the convergent $z/r$ that are required to achieve a high success probability when using continued fractions-based post-processing.
If the order is partially very smooth, it may however be necessary to perform an additional search in the classical post-processing to recover~$r$.
An algorithm for this search is given in~\cite[Sect.~6.2.4]{general}.

A takeaway from~\cite[App.~A]{general} is that the simulations show that a single run with~$\ell \sim m$ is usually sufficient to recover the order~$r$.
In this paper, we formalize this result, by proving a lower bound on the success probability that holds for any~$r$.

\subsection{Overview of our contribution}
\label{sec:overview-our-contribution}
As stated above, there are a number of works that follow up on Shor's original work to give better bounds on the success probability of the order-finding algorithm.

We are however not currently aware of any previous formal lower bound on the success probability that indicates that a single run of the quantum part is usually sufficient to recover the order~$r$ without increasing the exponent length, and that the success probability tends to one in the limit as~$r$ tends to infinity.

In this work, we prove such a bound, for the quantum part of the algorithm as described in Sect.~\ref{sec:quantum-order-finding}.
A key to achieving this result is to perform limited searches in the classical post-processing, and to account for these searches in the analysis.
As stated in Sect.~\ref{sec:shors-original-work}, Shor points to searching as an option in his original work.
Before giving an overview of our analysis, let us first revisit the exponent length:

Suppose that we use traditional continued fractions-based post-processing:
We then pick~$m$ and~$\ell$ such that $2^m > r$ and $2^{m+\ell} > r^2$.
The exponent, which is of length $m + \ell$~bits in this work, is then on par with or slightly shorter than the exponent in Shor's original work, depending on how~$r$ is bounded, see App.~\ref{appendix:notes-shor-order-finding-factoring}.

Suppose that we instead use lattice-based post-processing:
We may then pick a slightly smaller~$\ell$, without voiding the lower bound on the probability of successfully recovering~$r$ in a single run, provided that we accept to perform a limited enumeration of the lattice.
For further details, see the next section and App.~\ref{appendix:lattice-based-post-processing}.

\subsubsection{Overview of our analysis}
In the first part of our analysis, we sum up $P(\alpha_r)$ in a small $B$-neighborhood around the optimal frequency $j_0(z)$ for the peak with index~$z$ for $z \in [0, r) \inset \mathbb Z$.
Specifically, for $B \in [1, B_{\max}) \inset \mathbb Z$ where $B_{\max} = (2^{m+\ell}/r - 1)/2$, we lower-bound the sum
\begin{align*}
 S(z)
 =
 \sum_{t \, = \, -B}^{B} P(\{ r(j_0(z) + t) \}_{2^{m+\ell}})
 =
 \sum_{t \, = \, -B}^{B} P(\alpha_0(z) + rt)
\end{align*}
where we have used that for $\alpha_0(z) \in (-r/2, r/2]$ and $t \in [-B, B] \inset \mathbb Z$, it holds that
\begin{align}
  \{ r(j_0(z) + t) \}_{2^{m+\ell}}
  &=
  \{ \alpha_0(z) + rt \}_{2^{m+\ell}} \notag \\
  &=
  \alpha_0(z) + rt \in [-2^{m+\ell-1}, 2^{m+\ell-1}), \label{eq:requirement-met-by-Bmax}
\end{align}
enabling us to eliminate the modular reductions.

To this end, in Sect.~\ref{sec:approximation}, we first introduce an error-bounded approximation $\widetilde{P}(\alpha_r)$ of $P(\alpha_r)$.
In Sect.~\ref{sec:uniform}, we then use $\widetilde{P}(\alpha_r)$ to lower-bound $S(z)$.
By Thm.~\ref{th:P-a0-universal},
\begin{align*}
  S(z)
  \ge
  \frac{1}{r} \left(1 - \frac{1}{\pi^2} \left( \frac{2}{B} + \frac{1}{B^2} + \frac{1}{3B^3} \right) \right)
  -
  \frac{\pi^2 (2B+1)}{2^{m+\ell}}
\end{align*}
independent of~$z$.
This implies that if we try to solve not only the frequency~$j$ observed but also $j \pm 1$, $j \pm 2$, $\ldots$, $j \pm B$ for $z/r$, then with probability
\begin{align*}
  1 - \frac{1}{\pi^2} \left( \frac{2}{B} + \frac{1}{B^2} + \frac{1}{3B^3} \right)
  -
  \frac{\pi^2 r (2B+1)}{2^{m+\ell}}
\end{align*}
we will solve $j_0(z)$ for $z/r$, where critically~$z$ is uniformly distributed on $[0, r) \inset \mathbb Z$.

In the second part of our analysis, we first explain in Sect.~\ref{sec:find-order} that we will find $z/r$ --- and hence $\tilde r = r / d$ where $d = \gcd(r, z)$ --- if we pick~$\ell$ such that $2^{m+\ell} > r^2$ and solve $j_0(z)$ for $z/r$ by expanding $j_0(z) / 2^{m+\ell}$ in a continued fraction.

As an alternative approach, we furthermore show in Sect.~\ref{sec:find-order} and App.~\ref{appendix:lattice-based-post-processing} that we may instead pick $\ell = m - \Delta$ for some small $\Delta$ and use lattice-based techniques to solve~$j_0(z)$ for~$\tilde r$ at the expense of enumerating at most $\deltaexpr$ lattice vectors.

Next, we lower-bound the probability of~$d$ being $cm$-smooth\footnote{For the definition of $cm$-smooth used in this work, please see Sect.~\ref{sec:cm-smooth}.}, and we show that~$d$ and hence~$r$ can be efficiently recovered from~$\tilde r$ when~$d$ is $cm$-smooth.
We give algorithms for recovering either a multiple~$r'$ of~$r$ (see Alg.~\ref{alg:recover-multiple-of-r}), or~$r$ (see Alg.~\ref{alg:recover-r}--\ref{alg:recover-r-tree}), and for filtering the candidates for~$\tilde r$ generated in the post-processing (see Alg.~\ref{alg:filter-tilde-r}).

Finally, we wrap up the analysis in our main theorem, Thm.~\ref{th:main-theorem}:
Specifically, we lower-bound the probability of recovering~$r$ in a single run of the quantum part of the order-finding algorithm as a function of~$B$ and an additional parameter~$c$.
This when using either continued fractions-based or lattice-based post-processing.
In Cor.~\ref{cor:main-theorem-lattice} to Thm.~\ref{th:main-theorem} we give an analogous bound that holds when using lattice-based post-processing and enumerating a bounded number of lattice vectors.

In Cor.~\ref{cor:main-theorem-asymptotic} to Thm.~\ref{th:main-theorem}, we analyze the asymptotic behavior of the bound:
In particular, we show that for $B = m = \ordo(\poly(\log r))$ and $c = 1$, the success prob\-ability tends to one in the limit as~$r$ tends to infinity.
For these choices of parameters, the algorithm as a whole executes in polynomial time.\footnote{Note that a bound $m = \ordo(\poly(\log r))$ such that $r < 2^m$ must be known.}

Already for moderate $m = 128$, a high success probability exceeding e.g.~$1 - 10^{-4}$ can be guaranteed, see Tab.~\ref{tab:probability} where the bound is tabulated.\footnote{We could pick an even smaller value of~$m$, and still achieve a high success probability, but as~$m$ decreases below $128$ the order-finding problem starts to become classically tractable.}

\subsubsection{On the relation to our recent work on factoring completely}
In~\cite{completely}, it was recently shown that given the order~$r$ of a single element~$g$ selected uniformly at random from~$\mathbb Z_N^*$, the complete factorization of~$N$ may be recovered in classical polynomial time with very high probability of success.

This served as our motivation for proving in this paper that a single run of the quantum part of the order-finding algorithm is sufficient to find~$r$ given~$g$ with high probability --- and hence, by extension via~\cite{completely}, the complete factorization of~$N$.

In Cor.~\ref{cor:main-theorem-factor} and Cor.~\ref{cor:main-theorem-factor-lattice} in Sect.~\ref{sec:notes-on-factoring-via-order-finding}, we give lower bounds on the probability of completely factoring any integer~$N$ in a single order-finding run.
This when account\-ing for the probability of the order-finding algorithm failing to recover~$r$.

Furthermore, in Cor.~\ref{cor:main-theorem-asymptotic-factor} we show that the probability of completely factoring~$N$ in a single run tends to one in the limit as~$N$ tends to infinity.

\section{Approximating $P(\alpha_r)$ by $\widetilde{P}(\alpha_r)$}
\label{sec:approximation}
In this section, as a precursor to the main result, we upper-bound the error when approximating~$P(\alpha_r)$ by~$\widetilde{P}(\alpha_r)$, where
\begin{align*}
  \widetilde{P}(\alpha_r)
  =
  \frac{r}{2^{2(m+\ell)}}
  \frac{2(1 - \cos(2 \pi \alpha_r / r))}{(2 \pi \alpha_r / 2^{m+\ell})^2}.
\end{align*}

We derive the bound in two steps, by first upper-bounding the error when approx\-imating $P(\alpha_r)$ by $T(\alpha_r)$, and by then upper-bounding the error when approx\-imating $T(\alpha_r)$ by $\widetilde{P}(\alpha_r)$, where
\begin{align*}
  T(\alpha_r)
  =
  \frac{r}{2^{2(m+\ell)}}
  \frac{1 - \cos(2 \pi \alpha_r / r)}{1 - \cos(2 \pi \alpha_r / 2^{m+\ell})}.
\end{align*}

\subsection{Supporting claims}
\label{sec:approximations-supporting-claims}
To simplify our analysis, we introduce a few supporting claims in this section.
For the proofs of these claims, the reader is referred to App.~\ref{appendix:proofs-approximating}.

\begin{restatable}{claim}{mvtclaim}
  \label{claim:mvt}
  For any $u, v \in \mathbb R$, it holds that
  \begin{align*}
    |\, \cos(u) - \cos(v) \,| \le |\, u - v \,| \cdot \max\left( |\,u\,|, |\,v\,| \right).
  \end{align*}
\end{restatable}

\begin{restatable}{claim}{steponeclaimone}
\label{claim:step1-claim1}
It holds that
\begin{align*}
\left|\, \frac{2 \pi \alpha_r L}{2^{m+\ell}} \,\right| \le
\left|\, \frac{2 \pi \alpha_r (L + 1)}{2^{m+\ell}} \,\right| \le
\left|\, \frac{3 \pi \alpha_r}{r} \,\right|.
\end{align*}
\end{restatable}

\begin{restatable}{claim}{steponeclaimtwo}
\label{claim:step1-claim2}
It holds that
\begin{align*}
  \left|\, \frac{L+1}{2^{m+\ell}} - \frac{1}{r} \,\right|
  \le
  \frac{1}{2^{m+\ell}}
  \quad \text { and } \quad
  \left|\, \frac{L}{2^{m+\ell}} - \frac{1}{r} \,\right|
  <
  \frac{1}{2^{m+\ell}}.
\end{align*}
\end{restatable}

\begin{restatable}{claim}{steponeclaimcos}
\label{claim:cos}
For any $\phi \in [-\pi, \pi]$, it holds that
\begin{align*}
  \frac{2 \phi^2}{\pi^2} \le 1 - \cos \phi \le \frac{\phi^2}{2}.
\end{align*}
\end{restatable}

\begin{restatable}{claim}{steponeclaimcostwo}
\label{claim:cos2}
For any $\phi \in [-\pi, \pi]$, it holds that
\begin{align*}
  \left|\, (1 - \cos \phi) - \frac{\phi^2}{2} \,\right| \le \frac{\phi^4}{4!}.
\end{align*}
\end{restatable}

\subsection{Approximating $P(\alpha_r)$ by $T(\alpha_r)$}
\begin{lemma}
  \label{lemma:step1}
  The error when approximating $P(\alpha_r)$ by $T(\alpha_r)$ is bounded by
  \begin{align*}
    \left|\, P(\alpha_r) - T(\alpha_r) \,\right| < \frac{1}{2^{m+\ell}} \cdot \frac{3 \pi^2}{4}
    \quad \text{ when } \quad
    \alpha_r \neq 0.
  \end{align*}
\end{lemma}
\begin{proof}
First verify that
\begin{align}
       &\left|\, (1 - \cos(2 \pi \alpha_r (L + 1) / 2^{m+\ell})) - (1 - \cos(2 \pi \alpha_r / r)) \,\right| \notag \\
  = \, &\left|\, \cos \left( \frac{2 \pi \alpha_r (L + 1)}{2^{m+\ell}} \right) - \cos \left( \frac{2 \pi \alpha_r}{r} \right) \,\right| \notag \\
\le \, &\left|\, \frac{2 \pi \alpha_r(L + 1)}{2^{m+\ell}} - \frac{2 \pi \alpha_r}{r} \,\right| \cdot \left|\, \frac{3 \pi \alpha_r}{r} \,\right| \label{lemma:step1-approx1} \\
  = \, &\frac{6 \pi^2 \alpha^2_r}{r} \cdot \left|\, \frac{L + 1}{2^{m+\ell}} - \frac{1}{r} \,\right|
\le \,  \frac{6 \pi^2 \alpha^2_r}{2^{m+\ell} r} \label{lemma:step1-approx2}
\end{align}
where we have used Claim~\ref{claim:mvt} and~\ref{claim:step1-claim1} in~\refeq{lemma:step1-approx1}, and Claim~\ref{claim:step1-claim2} in~\refeq{lemma:step1-approx2}.
Analogously,
\begin{align}
  &\left|\, (1 - \cos(2 \pi \alpha_r L / 2^{m+\ell})) - (1 - \cos(2 \pi \alpha_r / r)) \,\right| \notag \\
= \, &\left|\, \cos \left( \frac{2 \pi \alpha_r L}{2^{m+\ell}} \right) - \cos \left( \frac{2 \pi \alpha_r}{r} \right) \,\right| \notag \\
\le \, &\left|\, \frac{2 \pi \alpha_r L}{2^{m+\ell}} - \frac{2 \pi \alpha_r}{r} \,\right| \cdot \left|\, \frac{3 \pi \alpha_r}{r} \,\right| \label{lemma:step1-approx1-2} \\
= \, &\frac{6 \pi^2 \alpha^2_r}{r} \cdot \left|\, \frac{L}{2^{m+\ell}} - \frac{1}{r} \,\right|
< \frac{6 \pi^2 \alpha^2_r}{2^{m+\ell} r} \label{lemma:step1-approx2-2}
\end{align}
where we have used Claim~\ref{claim:mvt} and~\ref{claim:step1-claim1} in~\refeq{lemma:step1-approx1-2}, and Claim~\ref{claim:step1-claim2} in~\refeq{lemma:step1-approx2-2}.

It now follows from~\refeq{lemma:step1-approx2} that
\begin{align}
  A
  &=
  \left|\,
  \frac{1 - \cos(2 \pi \alpha_r (L + 1) / 2^{m+\ell})}{1 - \cos(2 \pi \alpha_r / 2^{m+\ell})}
  -
  \frac{1 - \cos(2 \pi \alpha_r / r)}{1 - \cos(2 \pi \alpha_r / 2^{m+\ell})}
  \,\right| \notag \\
  &\le \,
  \frac{6 \pi^2 \alpha_r^2}{2^{m+\ell} r} \cdot
  \frac{1}{\left|\, 1 - \cos(2 \pi \alpha_r / 2^{m+\ell}) \,\right|} \notag \\
  &\le
  \frac{6 \pi^2 \alpha_r^2}{2^{m+\ell} r} \frac{\pi^2}{2 (2 \pi \alpha_r / 2^{m+\ell})^2}
  =
  \frac{2^{m+\ell}}{r} \cdot \frac{3 \pi^2}{4}, \label{lemma:step1-approx3}
\end{align}
where we have used Claim~\ref{claim:cos} in~\refeq{lemma:step1-approx3}.
Analogously, it follows from~\refeq{lemma:step1-approx2-2} that
\begin{align}
  B
  &=
  \left|\,
  \frac{1 - \cos(2 \pi \alpha_r L / 2^{m+\ell})}{1 - \cos(2 \pi \alpha_r / 2^{m+\ell})}
  -
  \frac{1 - \cos(2 \pi \alpha_r / r)}{1 - \cos(2 \pi \alpha_r / 2^{m+\ell})}
  \,\right| <
  \frac{2^{m+\ell}}{r} \cdot \frac{3 \pi^2}{4}. \label{lemma:step1-approx3-2}
\end{align}

Finally, it follows from~\refeq{eq:P-non-zero}, ~\refeq{lemma:step1-approx3} and~\refeq{lemma:step1-approx3-2} that
\begin{align*}
  \left|\, P(\alpha_r) - T(\alpha_r) \,\right|
  &=
  \frac{\beta}{2^{2(m+\ell)}} A + \frac{r-\beta}{2^{2(m+\ell)}} B
  <
  \frac{r}{2^{2(m+\ell)}} \cdot \frac{2^{m+\ell}}{r} \frac{3 \pi^2}{4}
  =
  \frac{1}{2^{m+\ell}} \cdot \frac{3 \pi^2}{4}
\end{align*}
and so the lemma follows.
\end{proof}

\subsection{Approximating $T(\alpha_r)$ by $\widetilde{P}(\alpha_r)$}
\begin{lemma}
  \label{lemma:step2}
  The error when approximating $T(\alpha_r)$ by $\widetilde{P}(\alpha_r)$ is bounded by
  \begin{align*}
    \left|\, T(\alpha_r) - \widetilde{P}(\alpha_r) \,\right| \le \frac{r}{2^{2(m+\ell)}}\frac{\pi^2}{12}
    \quad \text{ when } \quad
    \alpha_r \neq 0.
  \end{align*}
\end{lemma}
\begin{proof}
To condense the presentation, let us introduce $T'(\alpha_r)$ and $\widetilde{P}'(\alpha_r)$, where
\begin{align*}
  T(\alpha_r) &= \frac{r}{2^{2(m+\ell)}} \underbrace{\frac{1 - \cos(2 \pi \alpha_r / r)}{1 - \cos(2 \pi \alpha_r / 2^{m+\ell})}}_{T'(\alpha_r)},
  &
  \widetilde{P}(\alpha_r) &= \frac{r}{2^{2(m+\ell)}} \underbrace{\frac{2(1 - \cos(2 \pi \alpha_r / r))}{(2 \pi \alpha_r / 2^{m+\ell})^2}}_{\widetilde{P}'(\alpha_r)}.
\end{align*}

The lemma then follows from the fact that
\begin{align}
  \left|\, T'(\alpha_r) - \widetilde{P}'(\alpha_r) \,\right|
  &=
  \left|\,
  \frac{1 - \cos(2 \pi \alpha_r / r)}{1 - \cos(2 \pi \alpha_r / 2^{m+\ell})}
  -
  \frac{2(1 - \cos(2 \pi \alpha_r / r))}{(2 \pi \alpha_r / 2^{m+\ell})^2}
  \,\right| \notag \\
  &=
  \underbrace{\left|\, 1 - \cos(2 \pi \alpha_r / r) \,\right|}_{\le 2}
  \cdot
  \left|\,
  \frac{1}{1 - \cos \theta_r}
  -
  \frac{2}{\theta_r^2}
  \,\right| \notag \\
  &\le
  2
  \left|\,
  \frac{\theta_r^2 - 2(1 - \cos \theta_r)}{(1 - \cos \theta_r) \theta_r^2}
  \,\right|
  =
  4
  \left|\,
  \frac{\theta_r^2/2 - (1 - \cos \theta_r)}{(1 - \cos \theta_r) \theta_r^2}
  \,\right| \notag \\
  &\le
  4
  \left|\,
  \frac{\theta_r^4 / 4!}{(1 - \cos \theta_r) \theta_r^2}
  \,\right|
  =
  \frac{\theta_r^2}{6}
  \left|\,
  \frac{1}{1 - \cos \theta_r}
  \,\right| \label{lemma:step2-approx1} \\
  &\le
  \frac{\theta_r^2}{6}
  \cdot
  \frac{\pi^2}{2 \theta_r^2}
  =
  \frac{\pi^2}{12} \label{lemma:step2-approx2}
\end{align}
where we have used Claim~\ref{claim:cos2} in~\refeq{lemma:step2-approx1}, and Claim~\ref{claim:cos} in~\refeq{lemma:step2-approx2}.
\end{proof}

\subsection{Approximating $P(\alpha_r)$ by $\widetilde{P}(\alpha_r)$}
\begin{lemma}
\label{lemma:P}
The error when approximating $P(\alpha_r)$ by $\widetilde{P}(\alpha_r)$ is bounded by
\begin{align*}
  |\, P(\alpha_r) - \widetilde{P}(\alpha_r) \,| \le \tilde{\epsilon}
  =
  \frac{\pi^2}{2^{m+\ell}} \left( \frac{3}{4} + \frac{r}{2^{m+\ell}} \frac{1}{12} \right) < \frac{\pi^2}{2^{m+\ell}}
  \quad \text{ when } \quad
  \alpha_r \neq 0.
\end{align*}
\end{lemma}
\begin{proof}
By the triangle inequality, we have that
\begin{align*}
  \tilde{\epsilon} = |\, P(\alpha_r) - \widetilde{P}(\alpha_r) \,|
  &=
    |\, P(\alpha_r) - T(\alpha_r) + T(\alpha_r) - \widetilde{P}(\alpha_r) \,| \\
  &\le
  |\, P(\alpha_r) - T(\alpha_r) \,| + |\, T(\alpha_r) - \widetilde{P}(\alpha_r) \,|.
\end{align*}
The lemma then follows from Lem.~\ref{lemma:step1} and Lem.~\ref{lemma:step2}.
\end{proof}

\section{Proving approximate uniformity}
\label{sec:uniform}
In this section, we lower-bound the sum
\begin{align}
  \label{eq:sum-one-r}
  \sum_{t \, = \, -B}^{B} P(\alpha_0 + rt)
\end{align}
where $\alpha_0 \in (-r/2, r/2]$ and $B \in [1, B_{\max}) \inset \mathbb Z$ for $B_{\max} = (2^{m+\ell}/r - 1)/2$.

We first bound~\refeq{eq:sum-one-r} assuming $\alpha_0 \neq 0$ in Sect.~\ref{sec:uniform-a0-non-zero}.
We then extend this to a lower bound that holds for all $\alpha_0 \in (-r/2, r/2]$ in Sect~\ref{sec:uniform-a0-zero} and Sect.~\ref{sec:uniform-lower-bound}.

\subsection{The case $\alpha_0 \neq 0$}
\label{sec:uniform-a0-non-zero}
In what follows, let $\psi'(x)$ denote the trigamma function; the first derivative of the digamma function $\psi(x)$, or equivalently the second derivative of $\ln \Gamma(x)$.

Our analysis makes use of $\psi'(x)$, and in particular of the below supporting lemma and claim.
For their proofs, the reader is referred to App.~\ref{appendix:proofs-proving-uniformity}.

\begin{restatable}{lemma}{sumtrigamma}
  \label{lemma:sum-trigamma}
  For non-zero $\alpha_0 \in (-r/2, r/2]$, and~$B$ a positive integer,
  \begin{align*}
    &\sum_{t = -B}^{B}
    \frac{1}{(\alpha_0 + rt)^2}
    = \\
    &\quad\quad
    \frac{1}{r^2}
    \left(
      \frac{2 \pi^2}{1 - \cos(2 \pi \alpha_0 / r)} - (\psi'(1 + B + \alpha_0/r) + \psi'(1 + B - \alpha_0/r))
    \right).
  \end{align*}
\end{restatable}

\begin{restatable}{claim}{trigamma}
\label{claim:trigamma}
For any real $x > 0$,
\begin{align*}
  \psi'(x)
  <
  \frac{1}{x} + \frac{1}{2x^2} + \frac{1}{6x^3}.
\end{align*}
\end{restatable}

\noindent
We are now ready to proceed to lower-bound the sum in the non-zero case:

\begin{theorem}
  \label{th:P-a0-non-zero}
  For non-zero $\alpha_0 \in (-r/2, r/2]$, and $B \in [1, B_{\max}) \inset \mathbb Z$,
  \begin{align*}
    \sum_{t \,= \, -B}^{B} P(\alpha_0 + rt) = \frac{1}{r}(1-\epsilon_R) + \epsilon_A
    \text{ where }
    0 \le
    \epsilon_R
    \le
    \epsilon_{R,0}
    =
    \frac{1}{\pi^2}
    \left(
      \frac{2}{B} + \frac{1}{B^2} + \frac{1}{3B^3}
    \right)
  \end{align*}
  and $\left|\, \epsilon_A \,\right| \le \epsilon_{A,0} = (2B + 1) \tilde{\epsilon}$ for $\tilde{\epsilon}$ as in Lem.~\ref{lemma:P}.
\end{theorem}
\begin{proof}
  For the absolute error $\epsilon_A$, we have that
  \begin{align*}
    \sum_{t \, = \, -B}^{B} P(\alpha_0 + rt)
    =
    \sum_{t \, = \, -B}^{B} \widetilde{P}(\alpha_0 + rt) +
    \underbrace{\sum_{t \, = \, -B}^{B} P(\alpha_0 + rt) -
    \sum_{t \, = \, -B}^{B} \widetilde{P}(\alpha_0 + rt)}_{\epsilon_A}
  \end{align*}
  where, by the triangle inequality and Lem.~\ref{lemma:P}, it holds that
  \begin{align*}
    \left|\, \epsilon_A \,\right|
    &=
    \left|\, \sum_{t \, = \, -B}^{B} (P(\alpha_0 + rt) - \widetilde{P}(\alpha_0 + rt)) \,\right| \\
    &\le
    \sum_{t \, = \, -B}^{B} \left|\, P(\alpha_0 + rt) - \widetilde{P}(\alpha_0 + rt) \,\right|
    \le
    (2B + 1) \, \tilde{\epsilon}.
  \end{align*}

  For the relative error $\epsilon_R$, we have that
  \begin{align*}
    \sum_{t \, = \, -B}^{B} \widetilde{P}(\alpha_0 + rt)
    &=
    \frac{r}{2^{2(m+\ell)}}
    \sum_{t \, = \, -B}^{B}
    \frac{2(1 - \cos(2 \pi (\alpha_0 + rt) / r))}{(2 \pi (\alpha_0 + rt) / 2^{m+\ell})^2} \\
    &=
    \frac{(1 - \cos(2 \pi \alpha_0/r)) r}{2\pi^2}
    \sum_{t \, = \, -B}^{B}
    \frac{1}{(\alpha_0 + rt)^2}
    =
    \frac{1}{r} \left( 1 - \epsilon_R \right)
  \end{align*}
  where we have used that
  \begin{align*}
    &\sum_{t \, = \, -B}^{B}
    \frac{1}{(\alpha_0 + rt)^2}
    = \\
    &\quad\quad
    \frac{1}{r^2}
    \left(
    \frac{2 \pi^2}{1 - \cos(2\pi \alpha_0 / r)}
    -
    (
      \psi'(1 + B - \alpha_0 / r)
      +
      \psi'(1 + B + \alpha_0 / r)
    )
    \right)
  \end{align*}
  by Lem.~\ref{lemma:sum-trigamma}, which implies that
  \begin{align*}
  \epsilon_R
  &=
  \frac{1 - \cos(2\pi \alpha_0 / r)}{2 \pi^2}
  (
    \psi'(1 + B - \alpha_0 / r)
    +
    \psi'(1 + B + \alpha_0 / r)
  ) \\
  &\le
  \epsilon_{R,0}
  =
  \frac{1}{\pi^2}
  \left(
    \frac{2}{B} + \frac{1}{B^2} + \frac{1}{3B^3}
  \right)
  \end{align*}
  where we have used Claim~\ref{claim:trigamma}, and that $1 + B \pm \alpha_0/r \ge B$ as $|\, \alpha_0 \,| \le r/2$.
\end{proof}

\subsection{The special case $\alpha_0 = 0$}
\label{sec:uniform-a0-zero}
For completeness, we also treat the case where~$\alpha_0$ is zero in this section.

Note however that $\alpha_0 = 0 \Rightarrow j = 2^{m+\ell-\kappa_r} u$ for some integer $u \in [0, 2^{\kappa_r})$ and~$2^{\kappa_r}$ the greatest power of two to divide~$r$.
This in turn implies that $j / 2^{m+\ell} = u/2^{\kappa_r}$ which typically yields insufficient information to solve for~$r$ unless~$\kappa_r$ is very large.

\begin{lemma}
\label{lemma:uniform-a0-zero}
It holds that
\begin{align*}
  P(0) = \frac{1}{r} + \epsilon
  \quad \text{ where } \quad
  0 \le \epsilon = \frac{\beta (r-\beta)}{2^{2(m+\ell)} r} < \frac{1}{2^{m+2\ell}}.
\end{align*}
\end{lemma}
\begin{proof}
We have that $0 \le \beta = 2^{m+\ell} \text{ mod } r < r < 2^m$, and by~\refeq{eq:P-zero} that
\begin{align*}
  P(0)
  &=
  \frac{L^2 r + (2L + 1)\beta}{2^{2(m+\ell)}}
  =
  \frac{1}{2^{2(m+\ell)}}
  \left(
    \floor{\frac{2^{m+\ell}}{r}}^2 r
    +
    \left( 2 \floor{\frac{2^{m+\ell}}{r}} + 1 \right) \beta
  \right) \\
  &=
  \frac{1}{2^{2(m+\ell)}}
  \left(
    \left( \frac{2^{m+\ell}}{r} - \frac{\beta}{r} \right)^2 r
    +
    2 \left( \frac{2^{m+\ell}}{r} - \frac{\beta}{r} \right)\beta
    +
    \beta
  \right)
  =
  \frac{1}{r}
  +
  \frac{\beta (r-\beta)}{2^{2(m+\ell)} r},
\end{align*}
where $0 \le \beta(r-\beta) / (2^m r) < 1$, and so the lemma follows.
\end{proof}

\subsection{Lower-bounding the sum}
\label{sec:uniform-lower-bound}
\begin{theorem}
  \label{th:P-a0-universal}
  For $\alpha_0 \in (-r/2, r/2]$, and $B \in [1, B_{\max}) \inset \mathbb Z$,
  \begin{align*}
    \sum_{t \,= \, -B}^{B} P(\alpha_0 + rt)
    &\ge
    \frac{1}{r} \left(
      1
      -
      \frac{1}{\pi^2}
      \left(
        \frac{2}{B} + \frac{1}{B^2} + \frac{1}{3B^3}
      \right)
    \right)
    -
    \frac{\pi^2 (2B + 1)}{2^{m+\ell}}.
  \end{align*}
\end{theorem}
\begin{proof}
  For non-zero $\alpha_0 \in (-r/2, r/2]$, we have by Thm.~\ref{th:P-a0-non-zero} that
  \begin{align*}
    \sum_{t \,= \, -B}^{B} P(\alpha_0 + rt)
    &=
    \frac{1}{r}(1-\epsilon_R) + \epsilon_A
    \ge
    \frac{1}{r}(1-\epsilon_{R,0}) - \epsilon_{A,0} \\
    &\ge
    \frac{1}{r} \left(
      1
      -
      \frac{1}{\pi^2}
      \left(
        \frac{2}{B} + \frac{1}{B^2} + \frac{1}{3B^3}
      \right)
    \right)
    -
    \frac{\pi^2 (2B + 1)}{2^{m+\ell}}
  \end{align*}
  where we have used that $\left|\, \epsilon_A \,\right| \le \epsilon_{A,0} = (2B + 1) \tilde{\epsilon}$, where $\tilde \epsilon < \pi^2 / 2^{m+\ell}$ by Lem.~\ref{lemma:P}.

  For $\alpha_0 = 0$, we have that $P(\alpha_0) = P(0) \ge 1/r$ by Lem.~\ref{lemma:uniform-a0-zero}, and $P(\alpha_0 + rt) = P(rt) \ge 0$ for $|\, t \,| \in [1, B] \inset \mathbb Z$ as~$P$ is non-negative, and so the theorem follows.
\end{proof}

\section{Finding the order $r$ given $j$}
\label{sec:find-order}
As stated earlier in the introduction in Sect.~\ref{sec:shors-original-work}, Shor~\cite{shor94, shor97} originally proposed to recover~$r$ from~$j$ by expanding $j/2^{m+\ell}$ in a continued fraction.

By Claim~\ref{claim:convergent-existence} in App.~\ref{appendix:continued-fractions-based-post-processing}, it suffices to require that
\begin{align}
  \label{eq:hw-requirement}
  \left|\, \frac{j}{2^{m+\ell}} - \frac{z}{r} \,\right|
  <
  \frac{1}{2r^2}
\end{align}
for the convergent $z/r$ to appear in the continued fraction expansion of $j / 2^{m+\ell}$.

By Thm.~\ref{th:P-a0-universal}, it suffices with high probability to search a small $B$-neighborhood around the frequency observed to find the optimal frequency $j_0(z)$ that yields $\alpha_0(z) \in (-r/2, r/2]$.
For $j = j_0(z)$, we then have that
\begin{align}
  \label{eq:hw-requirement-j0}
  \left|\, \frac{j_0(z)}{2^{m+\ell}} - \frac{z}{r} \,\right|
  =
  \left|\, \frac{rj_0(z) - 2^{m+\ell} z}{2^{m+\ell} r} \,\right|
  =
  \frac{\left|\, \{r j_0(z)\}_{2^{m+\ell}} \,\right|}{2^{m+\ell} r}
  =
  \frac{\left|\, \alpha_0(z) \,\right|}{2^{m+\ell} r}
  \le
  \frac{1}{2 \cdot 2^{m+\ell}},
\end{align}
so by Claim~\ref{claim:convergent-existence} it then suffices that $2^{m+\ell} > r^2$ to recover the convergent $z/r$.
Note that as $2^m > r$, the requirement that $2^{m+\ell} > r^2$ is met for any~$\ell \ge m$.

Given $z/r$, we may immediately recover $\tilde r = r/d$ where $d = \gcd(r, z)$.
By also using Claim~\ref{claim:convergent-uniqueness} to identify $z/r$ in the expansion, we obtain the below lemma:

\begin{restatable}{lemma}{continuedfractionslemma}
  \label{lemma:continued-fractions-recover-tilde-r}
  The last convergent $p/q$ with denominator $q < 2^{(m+\ell)/2}$ in the continued fraction expansion of~$j / 2^{m+\ell}$ is equal to $z/r$, for $j = j_0(z)$ for any $z \in [0, r) \inset \mathbb Z$, and for $m, \ell \in \mathbb Z_{> 0}$ such that $2^m > r$ and $2^{m+\ell} > r^2$.
\end{restatable}

We may use lattice-based post-processing to achieve an analogous result:
More specifically, we may recover $\tilde r/2$ and hence~$\tilde r$ by using Lagrange's algorithm~\cite{lagrange, nguyen} to find the shortest non-zero vector, up to sign, of a two-dimensional lattice:
\begin{restatable}{lemma}{latticelemmashortest}
  \label{lemma:lattice-recover-tilde-r-shortest}
  The shortest non-zero vector, up to sign, in the lattice~$\mathcal L$ spanned by $(j, \frac{1}{2})$ and $(2^{m+\ell}, 0)$ has $\tilde r / 2 = r / (2 \gcd(r, z))$ in its second component, for $j = j_0(z)$ for any $z \in [0, r) \inset \mathbb Z$, and for $m, \ell \in \mathbb Z_{> 0}$ such that $2^m > r$ and $2^{m + \ell} > r^2$.
\end{restatable}

When using lattice-based post-processing, we may in fact select $\ell = m - \Delta$ for some $\Delta \in [0, m) \inset \mathbb Z$ and still recover $\tilde r = r / \gcd(r, z)$ by enumerating at most $\deltaexpr$ vectors in the lattice.
For small~$\Delta$, the enumeration is efficient:

\begin{restatable}{lemma}{latticelemma}
  \label{lemma:lattice-recover-tilde-r}
  At most $\deltaexpr$ vectors in the lattice~$\mathcal L$ spanned by $(j, \frac{1}{2})$ and $(2^{m+\ell}, 0)$ must be enumerated to recover $\tilde r = r / \gcd(r, z)$, for $j = j_0(z)$ for any $z \in [0, r) \inset \mathbb Z$, for $m \in \mathbb Z_{> 0}$ such that $2^m > r$, and for $\ell = m - \Delta$ for some $\Delta \in [0, m) \inset \mathbb Z$.

  More specifically, a set of at most $\deltaexpr$ candidates for~$\tilde r$, that is guaranteed to contain~$\tilde r$, may be constructed by enumerating at most~$\deltaexpr$ vectors $\vec w = (w_1, w_2) \in \mathcal L$ of norm $|\, \vec w \,| \le 2^{m-1/2}$ and including $2 w_2$ in the set.
\end{restatable}

For further details, and the proofs of Lem.~\ref{lemma:continued-fractions-recover-tilde-r}--\ref{lemma:lattice-recover-tilde-r}, see App.~\ref{appendix:continued-fractions-based-post-processing} and App.~\ref{appendix:lattice-based-post-processing}.
For notes on slightly improving the constant $6 \sqrt{3}$ in Lem.~\ref{lemma:lattice-recover-tilde-r}, see~App.~\ref{appendix:lattice-based-post-processing-improved-bound}.

\subsection{Recovering the order $r$ from $\tilde r$}
\label{sec:find-order-recover-r-from-tilde-r}
Given $\tilde r = r / d$ where $d = \gcd(r, z)$, and~$g$, we may recover~$r$ when~$d$ is $cm$-smooth.

In what follows, we first formalize the notion of $cm$-smoothness in Sect.~\ref{sec:cm-smooth} and lower-bound the probability of~$d$ being $cm$-smooth in Sect.\ref{sec:cm-smooth-lower-bound}.
We then give algorithms for recovering~$r$, or a multiple~$r'$ of~$r$, from $\tilde r = r / d$ when~$d$ is $cm$-smooth, in Sect.~\ref{sec:recovering-multiple-of-r-given-r-tilde} thru Sect.~\ref{sec:recovering-r-given-r-tilde}.
To avoid invoking these algorithms for all candidates for~$\tilde r$ that are generated when solving not only~$j$ but also $j \pm 1, \, \ldots, \, j \pm B$ for~$\tilde r$, we also give an algorithm for filtering out good candidates for~$\tilde r$ in Sect.~\ref{sec:filtering-tilde-r}.

In Sect.~\ref{sec:solving-candidate-set-for-r}, we put all of these components together to efficiently recover~$r$ or~$r'$ from a set of integers known to contain $\tilde r = r/d$ where~$d$ is $cm$-smooth.

\subsubsection{The notion of $cm$-smoothness}
\label{sec:cm-smooth}
Throughout this paper, an integer is said to be $cm$-smooth if and only if it is positive and not divisible by any prime power greater than~$cm$.

\subsubsection{Bounding the probability of $d$ being $cm$-smooth}
\label{sec:cm-smooth-lower-bound}
When using Thm.~\ref{th:P-a0-universal} to lower-bound the probability of observing~$j$ that belongs to the $B$-neighborhood of $j_0(z)$ for $z \in [0, r) \inset \mathbb Z$, the peak index~$z$ is uniformly distributed on $[0, r) \inset \mathbb Z$ as the bound is independent of $\alpha_0(z)$ and hence of~$z$.

The probability of~$d$ being $cm$-smooth may then be lower-bounded:

\begin{lemma}
\label{lemma:probability-cm-smooth-d-uniform-z}
For~$z$ selected uniformly at random from $[0, r) \inset \mathbb Z$,
the probability that no prime power greater than $cm$ divides $d = \gcd(r, z)$ is lower-bounded by
\begin{align*}
1 - \frac{1}{c \log cm},
\end{align*}
for~$m$ such that $r < 2^m$ and $c \ge 1$ a parameter that may be freely selected.
\end{lemma}
\begin{proof}
There are at most~$\log r / \log cm$ prime powers~$q^e > cm$ that divide~$r$ (where~$q$ is a distinct prime for each power).
For each such prime power~$q^e$, the probability of~$q^e$ dividing~$z$ is~$1/q^e < 1/(cm)$.

By using that $\log r < m$, and taking a union bound, we obtain
\begin{align*}
  \frac{\log r}{\log cm} \cdot \frac{1}{cm} < \frac{1}{c \log cm},
\end{align*}
and so the lemma follows.
\end{proof}

In what follows, we primarily think of~$c$ as a constant, although it is possible to let~$c$ depend on~$m$.
To ensure that the classical post-processing is efficient when~$c$ depends on~$m$, we require $c = \ordo(\text{poly}(m))$ and $m = \ordo(\poly(\log r))$.

\subsubsection{Recovering a multiple $r'$ of $r$ from $\tilde r = r/d$ when $d$ is $cm$-smooth}
\label{sec:recovering-multiple-of-r-given-r-tilde}
Alg.~\ref{alg:recover-multiple-of-r} recovers a positive integer multiple~$r'$ of $r = d \cdot \tilde r$ given~$\tilde r$ when~$d$ is $cm$-smooth:

\begin{breakablealgorithm}
\caption{Recovers a positive integer multiple~$r'$ of $r = d \cdot \tilde r$ given~$\tilde r$ when~$d$ is $cm$-smooth.
$\:$
{\bfseries \emph{Inputs:}} $\tilde r$, $g$, $c$, $m$.
$\:$
{\bfseries \emph{Returns:}} $r'$ or~$\neg$ to signal failure to recover~$r'$.}
\label{alg:recover-multiple-of-r}
\begin{pseudocode}
  \item If $\tilde r \not\in [1, 2^m) \inset \mathbb Z$:
  \begin{pseudocode}
    \item Return $\neg$ to signal failure to recover~$r'$.
  \end{pseudocode}

  \item Let $r' \leftarrow \tilde r$ and $x \leftarrow g^{\tilde r}$. \label{alg:recover-multiple-step:exp-init}
  \item For $q \in \mathcal{P}(cm)$:

  \emph{Note: Iterate over $\mathcal P(cm)$ in increasing order so that $q \leftarrow 2, \, 3, \, 5, \, \ldots$}

  \begin{pseudocode}
    \item If $x = 1$:
    \begin{pseudocode}
      \item Return~$r'$.
    \end{pseudocode}
    \item Let $e \leftarrow \lfloor \log_q(cm) \rfloor$.
    \item Let $x \leftarrow x^{q^e}$ and $r' \leftarrow r' \cdot q^e$. \label{alg:recover-multiple-step:exp1}
  \end{pseudocode}
  \item \label{alg:recover-multiple-step:return}
  If $x \neq 1$:
  \begin{pseudocode}
    \item Return $\neg$ to signal failure to recover~$r'$.
  \end{pseudocode}
  \item Return~$r'$.
\end{pseudocode}
\end{breakablealgorithm}
Above in Alg.~\ref{alg:recover-multiple-of-r}, and in what follows, $\mathcal P(B)$ denotes the set of primes $\le B$.

The runtime of Alg.~\ref{alg:recover-multiple-of-r} is dominated by a sequence of exponentiations, which amount to exponentiating~$g$.
By Lem.~\ref{lemma:exponent-length-alg-recover-multiple-of-r}, the total exponent length is $\ordo(cm)$~bits:

\begin{lemma}
\label{lemma:exponent-length-alg-recover-multiple-of-r}
Alg.~\ref{alg:recover-multiple-of-r} exponentiates~$g$ to an exponent of length $\ordo(cm)$~bits.
\end{lemma}
\begin{proof}
The number of primes in $\mathcal P(cm)$ is $\ordo(cm / \ln cm)$.
For each such prime, the exponent in step~\ref{alg:recover-multiple-step:exp1} is at most of length $\ceil{\log cm}$ bits since $q^e \le cm$.

The total exponent length in bits is hence at most
\begin{align*}
  m + \ceil{\log cm} \cdot \ordo(cm / \ln cm) = \ordo(cm)
\end{align*}
where we have also accounted for step~\ref{alg:recover-multiple-step:exp-init} exponentiating~$g$ to $\tilde r \in [1, 2^m) \inset \mathbb Z$.
\end{proof}

The exponent length is usually expected to be much shorter when~$d$ is $cm$-smooth, as the algorithm aborts as soon as $x = 1$ instead of running to completion.

\subsubsection{Recovering $r$ from $\tilde r = r/d$ when $d$ is $cm$-smooth}
\label{sec:recovering-r-given-r-tilde}
Alg.~\ref{alg:recover-r} recovers $r = d \cdot \tilde r$ given~$\tilde r$ when~$d$ is $cm$-smooth:

\begin{breakablealgorithm}
  \caption{Recovers $r = d \cdot \tilde r$ given~$\tilde r$ when~$d$ is $cm$-smooth. \\
  {\bfseries \emph{Inputs:}} $\tilde r$, $g$, $c$, $m$.
  $\:$
  {\bfseries \emph{Returns:}} $r$ or~$\neg$ to signal failure to recover~$r$.}
  \label{alg:recover-r}
  \begin{pseudocode}
    \item If $\tilde r \not\in [1, 2^m) \inset \mathbb Z$:
    \begin{pseudocode}
      \item Return $\neg$ to signal failure to recover~$r$.
    \end{pseudocode}

    \item Let $x \leftarrow g^{\tilde r}$. \label{alg:recover-order-step:exp-init}
    \item If $x = 1$:
    \begin{pseudocode}
      \item Return~$\tilde r$.
    \end{pseudocode}
    \item Let~$S$ be an empty stack.

    \item For $q \in \mathcal{P}(cm)$:

    \emph{Note: Iterate over $\mathcal P(cm)$ in increasing order so that $q \leftarrow 2, \, 3, \, 5, \, \ldots$}

    \begin{pseudocode}
      \item Let $e \leftarrow \lfloor \log_q(cm) \rfloor$.
      \item Push $(x, q, e)$ onto~$S$.
      \item Let $x \leftarrow x^{q^e}$. \label{alg:recover-order-step:exp1}
      \item If $x = 1$:
      \begin{pseudocode}
        \item Stop iterating and go to step~\ref{alg:recover-order-step:backtrack}.
      \end{pseudocode}
    \end{pseudocode}

    \item \label{alg:recover-order-step:backtrack}
    If $x \neq 1$:
    \begin{pseudocode}
      \item Return $\neg$ to signal failure to recover~$r$.
    \end{pseudocode}

    \item
    Let $d \leftarrow 1$.

    \item \label{alg:recover-order-step:pop}
    While~$S$ is not the empty stack, pop $(x, q, e)$ from~$S$:

    \begin{pseudocode}
      \item Let $x \leftarrow x^d$. \label{alg:recover-order-step:exp2}

      \item For $i \leftarrow 1, \, \ldots, \, e$:
      \begin{pseudocode}
        \item If $x = 1$:
        \begin{pseudocode}
          \item Stop iterating and go to step~\ref{alg:recover-order-step:pop}.
        \end{pseudocode}
        \item Let $x \leftarrow x^q$ and $d \leftarrow d \cdot q$. \label{alg:recover-order-step:exp3}
      \end{pseudocode}
    \end{pseudocode}

    \item Return $d \cdot \tilde r$.
  \end{pseudocode}
\end{breakablealgorithm}

The runtime of Alg.~\ref{alg:recover-r} is dominated by multiple sequences of exponentiations, which all amount to exponentiating~$g$.
By Lem.~\ref{lemma:exponent-length-alg-recover-r}, the total exponent length is $\ordo(cm^2 / \log cm)$~bits:

\begin{lemma}
\label{lemma:exponent-length-alg-recover-r}
Alg.~\ref{alg:recover-r} performs multiple exponentiations of~$g$.
The total exponent length in all of these exponentiations is $\ordo(cm^2 / \log cm)$~bits.
\end{lemma}
\begin{proof}
The number of primes in $\mathcal P(cm)$ is $\ordo(cm / \ln cm)$.
For each such prime:
In step~\ref{alg:recover-order-step:exp1}, the exponent is at most of length $\ceil{\log cm}$ bits since $q^e \le cm$.
In step~\ref{alg:recover-order-step:exp2}, it is at most of length~$m$ bits since $d \le r < 2^m$.
In step~\ref{alg:recover-order-step:exp3}, it is at most of length $\lfloor \log_q(cm) \rfloor \cdot \ceil{\log q}$ bits.
The total exponent length in bits is hence at most
\begin{align*}
  m + \left( \ceil{\log cm} + m + \lfloor \log_q(cm) \rfloor \cdot \ceil{\log q} \right) \cdot \ordo(cm / \ln cm) = \ordo(c m^2 / \log cm)
\end{align*}
where we have also accounted for step~\ref{alg:recover-order-step:exp-init} exponentiating~$g$ to $\tilde r \in [1, 2^m) \inset \mathbb Z$.
\end{proof}

The exponent length is usually expected to be much shorter when~$d$ is $cm$-smooth, as the algorithm aborts as soon as $x = 1$ instead of running to completion, and as~$d$ is typically much smaller than~$r$ in step~\ref{alg:recover-order-step:exp2} of the algorithm.

A better asymptotic worst-case runtime may be achieved by instead performing a binary tree search, as in Alg.~\ref{alg:recover-r-tree} below.
Note however that it will often be the case in practice that Alg.~\ref{alg:recover-r} outperforms Alg.~\ref{alg:recover-r-tree}, since~$d$ is likely to only have a few small prime factors.
When this is the case, $d$ is found more quickly by Alg.~\ref{alg:recover-r}.

\begin{breakablealgorithm}
  \caption{Recovers $r = d \cdot \tilde r$ given~$\tilde r$ when~$d$ is $cm$-smooth. \\
  {\bfseries \emph{Inputs:}} $\tilde r$, $g$, $c$, $m$.
  $\:$
  {\bfseries \emph{Returns:}} $r$ or~$\neg$ to signal failure to recover~$r$.}
  \label{alg:recover-r-tree}
  \begin{pseudocode}
    \item If $\tilde r \not\in [1, 2^m) \inset \mathbb Z$:
    \begin{pseudocode}
      \item Return $\neg$ to signal failure to recover~$r$.
    \end{pseudocode}
    \item Let {\sc recursive}($x$, $F = \{q_1, \, \ldots, \, q_l\}$) be the following function:
    \begin{pseudocode}
      \item If $l = 1$:
      \begin{pseudocode}
        \item Return $\{ (q_1, x)  \}$.
      \end{pseudocode}
      \item Let $F_L \leftarrow \{ q_1, \, \ldots, \, q_{\floor{l / 2}} \}$ and $F_R \leftarrow \{ q_{\floor{l / 2} + 1}, \, \ldots, \, q_l \}$.
      \item Let $x_L \leftarrow x^{d_L}$ and $x_R \leftarrow x^{d_R}$, where \label{alg:recover-r-tree-step:exp-1}
      \begin{align*}
        d_L = \prod_{q \in F_R} q^{\lfloor \log_q(cm) \rfloor}
        \quad\quad \text{ and } \quad\quad
        d_R = \prod_{q \in F_L} q^{\lfloor \log_q(cm) \rfloor}.
      \end{align*}
      \item Return $\text{\sc recursive}(x_L, \, F_L) \, \cup \, \text{\sc recursive}(x_R, \, F_R)$.
    \end{pseudocode}

    \item Let $x \leftarrow g^{\tilde r}$ and $d \leftarrow 1$. \label{alg:recover-r-tree-step:exp-init}

    \item Let $T \leftarrow \text{\sc recursive}(x, \, \mathcal P(cm)) = \{ (q_1, x_1), \, \ldots, \, (q_l, x_l) \}$.

    \item For $(q_i, x_i) \in T$:
    \begin{pseudocode}
      \item Let $e_i \leftarrow 0$ and $e_{i, \max} \leftarrow \lfloor \log_{q_i}(cm) \rfloor$.
      \item While $x_i \neq 1$:
      \begin{pseudocode}
        \item If $e_i = e_{i, \max}$:
        \begin{pseudocode}
          \item Return $\neg$ to signal failure to recover~$r$.
        \end{pseudocode}

        \item Let $x_i \leftarrow x_i^{q_i}$, $d \leftarrow d \cdot q_i$ and $e_i \leftarrow e_i + 1$. \label{alg:recover-r-tree-step:exp-2}
      \end{pseudocode}
    \end{pseudocode}

    \item Return $r = d \cdot \tilde r$.
  \end{pseudocode}
\end{breakablealgorithm}

The runtime of Alg.~\ref{alg:recover-r-tree} is dominated by multiple sequences of exponentiations, which all amount to exponentiating~$g$.
By Lem.~\ref{lemma:exponent-length-alg-recover-r-tree}, the total exponent length in all of these exponentiations is $\ordo(cm \log cm)$~bits:
\begin{lemma}
  \label{lemma:exponent-length-alg-recover-r-tree}
  Alg.~\ref{alg:recover-r-tree} performs multiple exponentiations of~$g$.
  The total exponent length in all of these exponentiations is $\ordo(cm \log cm)$~bits.
\end{lemma}
\begin{proof}
The number of primes in $\mathcal P(cm)$ is $\ordo(cm / \ln cm)$.

In the first part of the algorithm, a tree is traversed by calling {\sc recursive} to construct~$T$.
The exponent length in bits at each level of the tree is at most $\ceil{\log cm} \cdot \ordo(cm / \ln cm)$ as $q^e \le cm$.
There are at most $\ceil{\log cm}$ levels.

Hence, the total exponent length in bits in all invocations of step~\ref{alg:recover-r-tree-step:exp-1} is
\begin{align*}
  \ceil{\log cm}^2 \cdot \ordo(cm / \ln cm) = \ordo(cm \log cm).
\end{align*}

In the second part, for each of the $\ordo(cm / \ln cm)$ entries in~$T$, the exponent length in bits in step~\ref{alg:recover-r-tree-step:exp-2} is at most $\lfloor \log_q(cm) \rfloor \cdot \ceil{\log q} = \ordo(\log cm)$.
Hence, the total exponent length in bits in the second part of the algorithm is
\begin{align*}
  \lfloor \log_q(cm) \rfloor \cdot \ceil{\log q} \cdot \ordo(cm / \ln cm) = \ordo(cm),
\end{align*}
for a total exponent length in bits of
\begin{align*}
  m + \ordo(cm) + \ordo(cm \log cm) = \ordo(cm \log cm),
\end{align*}
where we have also accounted for step~\ref{alg:recover-r-tree-step:exp-init} exponentiating~$g$ to $\tilde r \in [1, 2^m) \inset \mathbb Z$.
\end{proof}

\subsubsection{Filtering candidates for $\tilde r = r / d$ when $d$ is $cm$-smooth}
\label{sec:filtering-tilde-r}
Let $\mathcal S = \{ \tilde r_1, \, \ldots, \, \tilde r_l \}$ be a set of~$l$ candidates for $\tilde r = r/d$ where~$d$ is $cm$-smooth.

Alg.~\ref{alg:filter-tilde-r} then returns the subset~$\mathcal S'$ consisting of all~$\tilde r_i$ in~$\mathcal S$ that are such that~$d_i \cdot \tilde r_i$ is a positive integer multiple of~$r$, for~$d_i$ a $cm$-smooth integer:

\begin{breakablealgorithm}
  \caption{Returns the subset $\mathcal S'$ consisting of all~$\tilde r_i$ in $\mathcal S = \{ \tilde r_1, \, \ldots, \, \tilde r_l \}$ that are such that $d_i \cdot \tilde r_i$ is a positive integer multiple of~$r$, for~$d_i$ a $cm$-smooth integer. \\
  {\bfseries \emph{Inputs:}} $\mathcal S$, $g$, $c$, $m$.
  $\:$
  {\bfseries \emph{Returns:}} $\mathcal S' \subseteq \mathcal S$.}
  \label{alg:filter-tilde-r}
  \begin{pseudocode}
    \item Let $x \leftarrow g^e$ where
    \begin{align*}
      e = \prod_{q \in \mathcal P(cm)} q^{\lfloor \log_q cm \rfloor}.
    \end{align*} \label{alg:filter-tilde-r-step:pre-compute}
    \item Let~$\mathcal S'$ be an empty set.
    \item For $\tilde r_i \in \mathcal S$:
    \begin{pseudocode}
      \item If $\tilde r_i \in [1, 2^m) \inset \mathbb Z$ and $x^{\tilde r_i} = 1$: \label{alg:filter-tilde-r-step:test}
      \begin{pseudocode}
        \item Add~$\tilde r_i$ to~$\mathcal S'$.
      \end{pseudocode}
    \end{pseudocode}
    \item Return~$\mathcal S'$.
  \end{pseudocode}
\end{breakablealgorithm}

Note that step~\ref{alg:filter-tilde-r-step:pre-compute} that pre-computes~$x$ depends only on~$g$, $c$ and~$m$.
The actual test of each candidate for~$\tilde r$ in~$\mathcal S$ is performed in step~\ref{alg:filter-tilde-r-step:test}.
The set~$\mathcal S$ typically becomes available incrementally.
It may be filtered incrementally by pre-computing~$x$ and then executing step~\ref{alg:filter-tilde-r-step:test} for each candidate as it becomes available.

As for the other algorithms in this section, the runtime of Alg.~\ref{alg:filter-tilde-r} is dominated by multiple sequences of exponentiations, which all amount to exponentiating~$g$.
By Lem.~\ref{lemma:exponent-length-alg-recover-tilde-r}, the total exponent length is $\ordo((c+l)m)$~bits:

\begin{lemma}
\label{lemma:exponent-length-alg-recover-tilde-r}
Alg.~\ref{alg:filter-tilde-r} performs multiple exponentiations of~$g$.
The total exponent length in all of these exponentiations is $\ordo((c + l)m)$~bits.
\end{lemma}
\begin{proof}
In the first part, the number of primes in $\mathcal P(cm)$ is $\ordo(cm / \ln cm)$.
For each such prime, the exponent is at most of length $\ceil{\log cm}$ bits. Hence
\begin{align*}
  \ordo(\log e) \le \ceil{\log cm} \cdot \ordo(cm / \ln cm) = \ordo(cm).
\end{align*}

In the second part, it holds that $\tilde r_i \in [0, 2^m) \inset \mathbb Z$ for all~$\tilde r_i$ in~$\mathcal S$ for which we perform exponentiations, and there are~$l$ entries in~$\mathcal S$.
The total exponent length is hence $\ordo((c + l)m)$ bits, and so the lemma follows.
\end{proof}

\subsubsection{Notes on optimizing the filtering of candidates for $\tilde r$}
\label{sec:notes-optimizing-the-filtering-step}
The basic filtering procedure in Alg.~\ref{alg:filter-tilde-r} may be optimized in various ways.

Suppose e.g.\ that a non-negative multiple~$\mu$ of~$r$ is known.
Then Alg.~\ref{alg:filter-tilde-r} may be optimized by not exponentiating~$x$ to~$\tilde r_i$ in step~\ref{alg:filter-tilde-r-step:test}, but rather to $\gcd(\tilde r_i, \mu)$.

Note that this does not affect the correctness of the algorithm:
In particular, the correct $\tilde r = r / d$, where $d = \gcd(r, z)$ is $cm$-smooth, will pass the test in step~\ref{alg:filter-tilde-r-step:test} when the optimization is applied, since $\tilde r = \gcd(\tilde r, \mu)$, and be included in the subset~$\mathcal S'$.
The same holds for any candidate~$\tilde r_i$ for~$\tilde r$ in~$\mathcal S$ that meets the requirement that $d_i \cdot \tilde r_i$ is a positive integer multiple of~$r$, for~$d_i$ a $cm$-smooth integer.

Note furthermore that as soon as a candidate~$\tilde r_i$ for~$\tilde r$ passes the test in step~\ref{alg:filter-tilde-r-step:test}, we know that $\tilde r_i \cdot e$ is a positive multiple of~$r$, and so we may update~$\mu$ by letting $\mu \leftarrow \gcd(\tilde r_i \cdot e, \mu)$, further tightening the filter.
Initially, we may let $\mu \leftarrow 0$.
The advantage of this optimization is hence that once a candidate for~$\tilde r_i$ that passes the test has been found, it becomes easier to test additional candidates, as it then often suffices to exponentiate $x$ to a smaller exponent, or not at all.

Other more obvious optimizations in practical implementations involve keeping track of candidates for~$\tilde r_i$ that have already passed the test, and of reduced candidates~$\gcd(\tilde r_i, \mu)$ that have already been dismissed, so as to avoid repeatedly testing candidates that have already been tested, or for which sufficient information has been accumulated to immediate dismiss or accept the candidate.

\subsubsection{Solving a set of candidates for $\tilde r = r/d$ for $r$ when $d$ is $cm$-smooth}
\label{sec:solving-candidate-set-for-r}
As previously explained, so as to achieve a high probability of solving the optimal frequency $j_0(z)$ closest to~$j$ for $\tilde r = r / d$ where $d = \gcd(r, z)$, we solve not only~$j$, but also $j \pm 1, \, \ldots, \, j \pm B$, for~$\tilde r$.
This yields a set~$\mathcal S = \{ \tilde r_1, \, \ldots, \, \tilde r_l \}$ of candidates for~$\tilde r$ guaranteed to contain~$\tilde r$ if $j_0(z)$ was amongst the frequencies solved for~$\tilde r$.

To recover~$r$ from~$\mathcal S$ when $\tilde r = r/d \in \mathcal S$ and~$d$ is $cm$-smooth, we first call Alg.~\ref{alg:filter-tilde-r} to filter the candidates in~$\mathcal S$.
This yields a subset $\mathcal S' \subseteq \mathcal S$ containing all $\tilde r_i \in \mathcal S$ that are such that $d_i \cdot \tilde r_i$ is a positive integer multiple of~$r$‚ for $d_i$ a $cm$-smooth integer.

For all candidates for~$\tilde r$ in~$\mathcal S'$, we then call either Alg.~\ref{alg:recover-r} or Alg.~\ref{alg:recover-r-tree}, and return the minimum of the candidates for~$r$ thus produced.
This yields~$r$.

To see why this is, note that $\tilde r \in \mathcal S'$ (since $\tilde r = r/d \in \mathcal S$ where $d$ is $cm$-smooth), and that a positive integer multiple of~$r$ is returned by both Alg.~\ref{alg:recover-r} and Alg.~\ref{alg:recover-r-tree} for all candidates for~$\tilde r$ that are in~$\mathcal S'$.
For~$\tilde r$, both algorithms return~$r$.
Hence, taking the minimum yields~$r$.

\subsubsection{Notes on optimizing the basic procedure for solving $\tilde r$ for $r$}
The basic procedure outlined in Sect.~\ref{sec:solving-candidate-set-for-r} may be optimized in various ways:

In particular, when using lattice-based post-processing and solving a range of offsets in~$j$ for~$\tilde r$, information computed when reducing the first lattice basis may be used to speed up the reduction of subsequent bases, see App.~\ref{appendix:lattice-based-post-solving-range-of-offsets-in-j-efficiently}.

When enumerating the lattice for large~$\Delta$, it is furthermore advantageous to filter the candidates for~$\tilde r$ as an integrated part of the enumeration.
The structure of the lattice may then be leveraged to speed up the filtering step, see App.~\ref{appendix:lattice-based-post-processing-filtering-candidates-efficiently}.

As previously explained in Sect.~\ref{sec:notes-optimizing-the-filtering-step}, the filtering step may also be optimized by keeping track of the multiples of~$r$ that become known as candidates for~$\tilde r_i$ pass the filtering step.
This information may be used to reduce subsequent candidates.

If only a positive integer multiple of~$r$ is sought, it suffices to post-process the first~$\tilde r_i$ that is inserted into~$\mathcal S'$ by Alg.~\ref{alg:filter-tilde-r} by calling either Alg.~\ref{alg:recover-multiple-of-r}, Alg.~\ref{alg:recover-r} or Alg.~\ref{alg:recover-r-tree}.
As explained in Sect.~\ref{sec:filtering-tilde-r}, the candidates for~$\tilde r$ may be filtered incrementally by Alg.~\ref{alg:filter-tilde-r} as they become available, enabling the search for~$\tilde r$ to be aborted early.

\subsubsection{Notes on heuristic optimizations of the basic procedure}
Let us furthermore briefly describe two heuristic optimizations:

As soon as a candidate for~$\tilde r$ passes the filter, we may use it to compute a positive multiple~$r'$ of~$r$ by using one of Alg.\ref{alg:recover-multiple-of-r}--\ref{alg:recover-r-tree}.
If $r' / r$ is smooth, as is heuristically likely to be the case in practice, $r$ may then be found by e.g.\ using trial division.\footnote{For all primes~$q$ up to some bound: For as long as~$q$ divides~$r'$ and $g^{r'/q} = 1$: Let $r' \leftarrow r' / q$.}

When solving a range of offsets in~$j$ for~$\tilde r$, it is heuristically likely to be the case that a contiguous subrange yields candidates for~$\tilde r$ that pass the filter.
By identifying this subrange whilst solving we may typically reduce the search space.

Note that the two above optimizations are heuristic:
They void the lower bound on the success probability that we will derive next in Sect.~\ref{sec:lower-bound-success-probability}.
We nevertheless mention them briefly here as they tend to produce good results in practice.

\subsubsection{Notes on the efficiency of the post-processing algorithms}
Note that it follows from Lem.~\ref{lemma:exponent-length-alg-recover-multiple-of-r}--\ref{lemma:exponent-length-alg-recover-tilde-r} that Alg.~\ref{alg:recover-multiple-of-r}--\ref{alg:filter-tilde-r}, respectively, execute in poly\-nomial time, assuming that $c, l = O(\poly(m))$ and $m = O(\poly(\log r))$.

\subsection{Lower-bounding the success probability}
\label{sec:lower-bound-success-probability}
We are now ready to wrap up our analysis in the below main theorem:

\begin{theorem}
\label{th:main-theorem}
The quantum algorithm in combination with the classical continued fractions-based or lattice-based post-processing successfully recovers~$r$ in a single run with probability at least
\begin{align}
  \left(
    1
    -
    \frac{1}{\pi^2}
    \left(
      \frac{2}{B} + \frac{1}{B^2} + \frac{1}{3B^3}
    \right)
  -
  \frac{\pi^2 r (2B + 1)}{2^{m+\ell}}
  \right)
  \left(
    1 - \frac{1}{c \log cm}
  \right)
  \label{eq:main-theorem-probability}
\end{align}
for $m, \ell \in \mathbb Z_{>0}$ such that $2^m > r$ and $2^{m+\ell} > r^2$, $c \ge 1$, and $B \in [1, B_{\max}) \inset \mathbb Z$.
\end{theorem}
\begin{notes}
To remove the dependency on~$r$, use e.g.\ that $r / 2^{m+\ell} < 1 / \sqrt{2^{m+\ell}}$.
\end{notes}
\begin{proof}
By Thm.~\ref{th:P-a0-universal}, the probability
\begin{align*}
  \sum_{t \,= \, -B}^{B} P(\alpha_0(z) + rt)
  &\ge
  \frac{1}{r} \left(
    1
    -
    \frac{1}{\pi^2}
    \left(
      \frac{2}{B} + \frac{1}{B^2} + \frac{1}{3B^3}
    \right)
  \right)
  -
  \frac{\pi^2 (2B + 1)}{2^{m+\ell}}
\end{align*}
independent of the peak index $z \in [0, r) \inset \mathbb Z$.

Hence, we observe some~$j$ such that $|\, j - j_0(z) \,| \le B$, with~$z$ selected uniformly at random from $[0, r) \inset \mathbb Z$, with probability at least
\begin{align}
  1
  -
  \frac{1}{\pi^2}
  \left(
    \frac{2}{B} + \frac{1}{B^2} + \frac{1}{3B^3}
  \right)
  -
  \frac{\pi^2 r (2B + 1)}{2^{m+\ell}}.
  \label{eq:main-theorem-probability-j}
\end{align}

Assuming that we pass not only~$j$, but also $j \pm 1, \, \ldots, \, j \pm B$, to the post-processing solver, we will hence pass it $j_0(z)$ with at least the probability in~\refeq{eq:main-theorem-probability-j}.

By~Lem.~\ref{lemma:continued-fractions-recover-tilde-r}, passing $j_0(z)$ to the continued fractions-based solver, which expands $j_0(z) / 2^{m+\ell}$ in a continued fraction and returns the last convergent $p/q$ with denominator $q < 2^{(m+\ell)/2}$, yields $z/r$ and hence $\tilde r = r/d$ where $d = \gcd(r, z)$.

By Lem.~\ref{lemma:lattice-recover-tilde-r-shortest}, passing $j_0(z)$ to the lattice-based solver, so that~$j_0(z)$ is used to setup the basis for the lattice~$\mathcal L$ which is then Lagrange-reduced, yields~$\tilde r/2$, and hence~$\tilde r$, up to sign, as the second component of the shortest non-zero vector in~$\mathcal L$.

The probability of $d = \gcd(r, z)$ being $cm$-smooth is at least $1 - 1 / (c \log cm)$ by Lem.~\ref{lemma:probability-cm-smooth-d-uniform-z}.
If~$d$ is $cm$-smooth, Alg.~\ref{alg:recover-r} and Alg.~\ref{alg:recover-r-tree} are both guaranteed to return~$r$ given $\tilde r = r / d$.
If instead passed an incorrect candidate for~$\tilde r$, these algorithms will either signal a failure to find~$r$, or return some positive integer multiple of~$r$.

Hence, with probability at least as stated in the theorem, we may recover~$r$ by taking the minimum of the candidates for~$r$ produced by Alg.~\ref{alg:recover-r} or Alg.~\ref{alg:recover-r-tree} when this algorithm is passed the candidates for~$\tilde r$ produced when solving not only~$j$, but also $j \pm 1, \, \ldots, \, j \pm B$, for~$\tilde r$, using either the continued factions-based or the lattice-based post-processing, and so the theorem follows.
\end{proof}

\begin{notes}
As explained in Sect.~\ref{sec:solving-candidate-set-for-r}, the procedure in the proof may be optimized by filtering the candidates for~$\tilde r$ using Alg.~\ref{alg:filter-tilde-r}, before passing them to Alg.~\ref{alg:recover-r} or Alg.~\ref{alg:recover-r-tree}.
\end{notes}

\vspace{2mm}

In order for the quantum algorithm to run in polynomial time, we must require $m = O(\poly(\log r))$.
In practice, we would select the least~$m$ that we can guarantee meets the requirements that $2^m > r$, and then the least~$\ell$ such that $2^{m+\ell} > r^2$.

For the classical post-processing to run in polynomial time, we must similarly require $B, c = O(\poly(\log r))$.
As may be seen in Thm.~\ref{th:main-theorem}, it does however suffice to let~$B$ and~$c$ be constants to achieve a high success probability.
Already for moderate~$m$, a high success probability exceeding e.g.~$1 - 10^{-4}$ can be guaranteed, see Tab.~\ref{tab:probability} where the bound is tabulated in~$c$ and~$B$ for $m = \ell = 128$.

\begin{table}
  \begin{center}
    \begin{tabular}{cr|r|r|r|r|r|r}
      \multicolumn{2}{c}{} & \multicolumn{6}{|c}{$B$} \\
          &      &       1 &      10 &     100 &    1000 &  $10^4$ &  $10^5$ \\
      \hline
          &    1 & 0.56765 & 0.83887 & 0.85539 & 0.85696 & 0.85712 & 0.85714 \\
          &   10 & 0.65584 & 0.96920 & 0.98829 & 0.99011 & 0.99029 & 0.99030 \\
          &   25 & 0.65998 & 0.97532 & 0.99453 & 0.99636 & 0.99654 & 0.99656 \\
      $c$ &  100 & 0.66177 & 0.97797 & 0.99723 & 0.99906 & 0.99924 & 0.99926 \\
          &  250 & 0.66208 & 0.97842 & 0.99769 & 0.99953 & 0.99971 & 0.99973 \\
          &  500 & 0.66217 & 0.97856 & 0.99783 & 0.99967 & 0.99985 & 0.99987 \\
          & 1000 & 0.66222 & 0.97863 & 0.99790 & 0.99973 & 0.99992 & 0.99993 \\
      \hline
    \end{tabular}
  \end{center}
  \caption{The lower bound on the success probability in Thm.~\ref{th:main-theorem} rounded down and tabulated in $B$ and $c$ for $m = \ell = 128$.
  Further increasing~$m, \ell$ only increases the probability.
  When tabulating the bound, we used that $r / 2^{m+\ell} < 1/\sqrt{2^{m+\ell}}$ to remove the dependency on~$r$, see the note in Thm.~\ref{th:main-theorem}.}
  \label{tab:probability}
\end{table}

When using lattice-based post-processing, it is possible to reduce~$\ell$ slightly, at the expense of enumerating a bounded number of vectors in the lattice:

\begin{theoremcorollary}
  \label{cor:main-theorem-lattice}
  The quantum algorithm in combination with the classical lattice-based post-processing successfully recovers~$r$ in a single run with probability at least
  \begin{align}
    \left(
      1
      -
      \frac{1}{\pi^2}
      \left(
        \frac{2}{B} + \frac{1}{B^2} + \frac{1}{3B^3}
      \right)
    -
    \frac{\pi^2 r (2B + 1)}{2^{m+\ell}}
    \right)
    \left(
      1 - \frac{1}{c \log cm}
    \right)
    \label{eq:cor-main-theorem-probability}
  \end{align}
  for $m \in \mathbb Z_{>0}$ such that $2^m > r$, $\ell = m - \Delta$ for some $\Delta \in [0, m) \inset \mathbb Z$, $c \ge 1$, and $B \in [1, B_{\max}) \inset \mathbb Z$, provided that at most $\deltaexpr$ lattice vectors are enumerated for each of the at most $2B + 1$ candidate frequencies.
\end{theoremcorollary}
\begin{notes}
To remove the dependency on~$r$, use e.g.\ that $r / 2^{m+\ell} < 1 / 2^{\ell}$.
\end{notes}
\begin{proof}
  The proof follows from Thm.~\ref{th:main-theorem} and its proof, by using lattice-based post-processing and enumerating vectors in the lattice~$\mathcal L$:

  By Lem.~\ref{lemma:lattice-recover-tilde-r} and its proof, when $2^m > r$ and $\ell = m - \Delta$, at most $\deltaexpr$ vectors in~$\mathcal L$ must be enumerated to recover~$\tilde r$, assuming $j = j_0(z)$ is passed to the solver.
  More specifically, the enumeration then generates a set of at most $\deltaexpr$ candidates for~$\tilde r$ that is guaranteed to contain~$\tilde r$.
  If $j \neq j_0(z)$ is passed to the solver, the enumeration may be aborted after $\deltaexpr$ vectors have been enumerated.\footnote{If the procedure in the proof of Lem.~\ref{lemma:lattice-recover-tilde-r} is followed and~$\mathcal L$ enumerated only if $\lambda_2^{\perp} < 2^{m-\frac{1}{2}}$, then the enumeration generates a set of at most $\deltaexpr$ candidates for~$\tilde r$ even if $j \neq j_0(z)$, so there is then no need to abort the enumeration after $\deltaexpr$ vectors have been enumerated.}

  As in the proof of Thm.~\ref{th:main-theorem}:
  Not only~$j$ but also $j \pm 1$, $\ldots$, $j \pm B$ are passed to the solver, so as to guarantee a given minimum probability~\refeq{eq:main-theorem-probability-j} of passing $j_0(z)$ to the solver.
  At most $2B+1$ enumerations of at most $\deltaexpr$ lattice vectors are hence performed, where each vector yields at most one candidate for~$\tilde r$.

  The probability of $d = \gcd(r, z)$ being $cm$-smooth is then at least $1 - 1 / (c \log cm)$ by Lem.~\ref{lemma:probability-cm-smooth-d-uniform-z}.
  If~$d$ is $cm$-smooth, Alg.~\ref{alg:recover-r} and Alg.~\ref{alg:recover-r-tree} are both guaranteed to return~$r$ given $\tilde r = r / d$.
  If instead passed an incorrect candidate for~$\tilde r$, these algorithms will either  signal a failure to find~$r$, or return some positive integer multiple of~$r$.

  Hence, with probability at least as stated in the corollary, we may recover~$r$ by taking the minimum of the candidates for~$r$ produced by Alg.~\ref{alg:recover-r} or Alg.~\ref{alg:recover-r-tree} when this algorithm is passed the candidates for~$\tilde r$ produced when solving not only~$j$, but also $j \pm 1, \, \ldots, \, j \pm B$, for~$\tilde r$, by using lattice-based processing and enumerating at most $\deltaexpr$ vectors in~$\mathcal L$, and so the corollary follows.
\end{proof}

\begin{notes}
  As explained in App.~\ref{appendix:lattice-based-post-processing-filtering-candidates-efficiently}, the procedure in the proof may be optimized by filtering the candidates for~$\tilde r$ as an integrated a part of enumerating~$\mathcal L$.
  Furthermore, as explained in App.~\ref{appendix:lattice-based-post-processing-improved-bound}, the constant $6 \sqrt{3}$ may be improved.
\end{notes}

\vspace{2mm}

It is furthermore interesting to note that asymptotically in the limit as~$r$ tends to infinity, the probability of successfully recovering~$r$ tends to one:

\begin{theoremcorollary}
  \label{cor:main-theorem-asymptotic}
  In the limit as~$r$ tends to infinity, the probability of the quantum algo\-rithm in combination with the classical post-processing succeeding in recovering~$r$ in a single run tends to one.
  All algorithms involved may be parameterized so as to achieve this limit whilst executing in polynomial time.
  This assuming that and upper bound $m = \ordo(\poly(\log r))$ such that $2^m > r$ is known, and that the group arithmetic is efficient.
\end{theoremcorollary}
\begin{proof}
  The proof follows from Thm.~\ref{th:main-theorem}, by taking the limit of~\refeq{eq:main-theorem-probability} as~$r$ tends to infinity, with e.g.\ $c = 1$ and $B = m$.
  Other choices of~$c$ and~$B$ are possible.
\end{proof}

\subsubsection{Notes on factoring integers via order finding}
\label{sec:notes-on-factoring-via-order-finding}
One of our motivations for pursuing this work is the classical probabilistic algorithm recently introduced in~\cite{completely}:
It completely factors any integer~$N$ efficiently, and with a very high probability of success~\cite[Thm.~1]{completely}, given the order~$r$ of a single element~$g$ selected uniformly at random from~$\mathbb Z_N^*$.
The order~$r$ may for instance be computed in a single successful run of Shor's order-finding algorithm as in this work.

To connect~\cite[Thm.~1]{completely} to the lower bound in Thm.~\ref{th:main-theorem}, we first need to introduce a supporting claim.
For its proof, the reader is referred to App.~\ref{appendix:support-claims-carmichael}.
\begin{restatable}{claim}{carmichaelclaim}
  \label{claim:carmichael}
  For~$N$ an odd positive integer that is the product of $n \ge 2$ distinct prime factors, it holds for the Carmichael function~$\lambda$ that $\lambda(N) < 2^{1-n} N$.
\end{restatable}
Note that the idea of using~$N/2$ to bound~$r$ from above when factoring via order finding was seemingly first introduced by Gerjuoy~\cite{gerjuoy}, see Sect.~\ref{sec:earlier-works}.

We may now lower-bound the overall success probability of the factoring algo\-rithm in~\cite{completely} succeeding in completely factoring~$N$ in a single order-finding run:
\begin{theoremcorollary}
\label{cor:main-theorem-factor}
Let~$N$ be an odd $l$-bit integer with $n \ge 2$ distinct prime factors.
Assume that we select~$g$ uniformly at random from~$\mathbb Z_N^*$, attempt to compute the order~$r$ of~$g$ in a single run of the quantum algorithm as described in this work with continued fractions-based or lattice-based post-processing, and then attempt to completely factor~$N$ given~$r$ using the classical algorithm in~\cite{completely}.
Then the probability of recovering the complete factorization of~$N$ is at least
\begin{align}
  &\left(
    1
    -
    \frac{1}{\pi^2}
    \left(
      \frac{2}{B} + \frac{1}{B^2} + \frac{1}{3B^3}
    \right)
  -
  \frac{\pi^2 (2B + 1)}{\sqrt{2^{m+\ell}}}
  \right) \cdot \notag \\
  &\quad\quad
  \left(
    1 - \frac{1}{c \log cm}
  \right)
  \left(
    1 - 2^{-k} \, {n \choose 2} - \frac{1}{2 \varsigma^{2} \log^2 \varsigma l}
  \right)
  \label{eq:main-theorem-factor-probability}
\end{align}
for $m = l - 1$, $\ell$ the least integer such that $2^{m+\ell} \ge N^2/4$, $B \in [1, B_{\max}) \inset \mathbb Z$, $c \ge 1$, $\varsigma \ge 1$ a constant that may be freely selected (denoted~$c$ in~\cite{completely}), and $k \ge 1$ the number of iterations in the classical post-processing in~\cite{completely}.
\end{theoremcorollary}
\begin{notes}
To pick~$k$, use e.g.\ that $k \ge 2 \log n - 1 + \tau$ for some positive $\tau$ implies $2^{-k} \, {n \choose 2} \le 2^{-\tau}$.
To remove the dependency on~$n$, use e.g.\ that $n < l$ as $N < 2^l$.
\end{notes}
\begin{proof}
By Claim~\ref{claim:carmichael}, we have that $r \le \lambda(N) < N/2 < 2^{l-1}$.
We select $m = l - 1$ to ensure that $2^m > r$, and the least~$\ell$ such that $2^{m+\ell} \ge N^2/4 > r^2$.
The first factor in the lower bound~\refeq{eq:main-theorem-factor-probability} then follows from Thm.~\ref{th:main-theorem}, as $r / 2^{m+\ell} < 1 / \sqrt{2^{m+\ell}}$.
The second factor follows from~\cite[Thm.~1]{completely}, and so the corollary follows.
\end{proof}

Note that as is explained in~\cite[Sect.~1]{completely}, the requirement that $n \ge 2$ does not imply a loss of generality, since we may reduce perfect powers and test for primality in classical polynomial time without resorting to order finding.
Similarly, the requirement that~$N$ must be odd may be handled by performing trial division for small factors before attempting to factor~$N$ using more elaborate methods.
Hence, a single order-finding run suffices for any integer~$N$ with probability at least~\refeq{eq:main-theorem-factor-probability} as the cases excluded in Cor.~\ref{cor:main-theorem-factor} can be handled efficiently classically.

Again, we may instead pick $\ell = m - \Delta$ for some small $\Delta$ when using lattice-based post-processing and enumerating a bounded number of vectors in the lattice:

\begin{theoremcorollary}
  \label{cor:main-theorem-factor-lattice}
  Let~$N$ be an odd $l$-bit integer with $n \ge 2$ distinct prime factors.
  Assume that we select~$g$ uniformly at random from~$\mathbb Z_N^*$, attempt to compute the order~$r$ of~$g$ in a single run of the quantum algorithm as described in this work with lattice-based post-processing, and then attempt to completely factor~$N$ given~$r$ using the classical algorithm in~\cite{completely}.
  Then the probability of recovering the complete factorization of~$N$ is at least
  \begin{align}
    &\left(
      1
      -
      \frac{1}{\pi^2}
      \left(
        \frac{2}{B} + \frac{1}{B^2} + \frac{1}{3B^3}
      \right)
    -
    \frac{\pi^2 (2B + 1)}{2^{\ell}}
    \right) \cdot \notag \\
    &\quad\quad
    \left(
      1 - \frac{1}{c \log cm}
    \right)
    \left(
      1 - 2^{-k} \, {n \choose 2} - \frac{1}{2 \varsigma^{2} \log^2 \varsigma l}
    \right)
    \label{eq:main-theorem-factor-lattice-probability}
  \end{align}
  for $m = l - 1$, $\ell = m - \Delta$ for some $\Delta \in [0, m) \inset \mathbb Z$, $B \in [1, B_{\max}) \inset \mathbb Z$, $c \ge 1$, $\varsigma \ge 1$ a constant that may be freely selected (denoted~$c$ in~\cite{completely}), and $k \ge 1$ the number of iterations in the classical post-processing in~\cite{completely}, provided that at most $\deltaexpr$ lattice vectors are enumerated for each of the at most $2B + 1$ candidate frequencies.
\end{theoremcorollary}
\begin{notes}
To pick~$k$, use e.g.\ that $k \ge 2 \log n - 1 + \tau$ for some positive $\tau$ implies $2^{-k} \, {n \choose 2} \le 2^{-\tau}$.
To remove the dependency on~$n$, use e.g.\ that $n < l$ as $N < 2^l$.
\end{notes}
\begin{proof}
  By Claim~\ref{claim:carmichael}, we have that $r \le \lambda(N) < N/2 < 2^{l-1}$.
  We select $m = l - 1$ to ensure that $2^m > r$, and $\ell = m - \Delta$.
  The first factor in the lower bound~\refeq{eq:main-theorem-factor-lattice-probability} then follows from Cor.~\ref{cor:main-theorem-lattice}, as $r / 2^{m+\ell} < 2^{\ell}$.
  The second factor follows from~\cite[Thm.~1]{completely}, and so the corollary follows.
\end{proof}

As for the asymptotic behavior of the combined factoring algorithm, the overall probability of successfully recovering the complete factorization of~$N$ in a single order-finding run may be seen to tend to one in the limit as~$N$ tends to infinity:

\begin{theoremcorollary}
  \label{cor:main-theorem-asymptotic-factor}
  Let~$N$ be an odd $l$-bit integer with $n \ge 2$ distinct prime factors.

  Assume that we select~$g$ uniformly at random from~$\mathbb Z_N^*$, attempt to compute the order~$r$ of~$g$ in a single run of the quantum algorithm as described in this work, and then attempt to completely factor~$N$ given~$r$ using the classical algorithm in~\cite{completely}.

  Then, in the limit as~$N$ tends to infinity, the probability of successfully recovering the complete factorization of~$N$ tends to one.
  All algorithms involved may be parameterized so as to achieve this limit whilst executing in polynomial time in~$l$.
\end{theoremcorollary}
\begin{proof}
  Assume that we e.g.\ select~$m$ and~$\ell$ as in Cor.~\ref{cor:main-theorem-factor}, and $c = \varsigma = 1$, $B = m$ and $k = \ceil{3 \log l}$.
  Then, in the limit as~$N$ tends to infinity, the lower bound~\refeq{eq:main-theorem-factor-probability} in Cor.~\ref{cor:main-theorem-factor} tends to one.
  In particular, note that $n < l$ as $N < 2^l$, and that
  \begin{align*}
    2^{-k} \cdot {n \choose 2}
    =
    2^{-\ceil{3 \log l}} \cdot \frac{n(n-1)}{2}
    \le
    \frac{l(l-1)}{2l^3}
    \le
    \frac{1}{l}.
  \end{align*}

  For these choices of parameters, all algorithms involved execute in polynomial time in~$l$, and so the corollary follows.
\end{proof}

Again, for the reasons explained above, the bound in Cor.~\ref{cor:main-theorem-asymptotic-factor} extends to any integer~$N$, as the cases excluded in Cor.~\ref{cor:main-theorem-asymptotic-factor} can be handled efficiently classically.

\subsubsection{Notes on the bounds underestimating the success probability}
When deriving the bounds in this work, we have assumed that the least~$m$ for which it can be guaranteed that $2^m > r$ is selected, and that a lower bound on the success probability is sought given only the guarantee that $r < 2^m$.
Similarly, we have assumed that the least~$\ell$ that fulfills the requirements on~$\ell$ --- e.g.~that $2^{m+\ell} > r$ or that $\ell = m - \Delta$ --- is selected.
If~$m$ or~$\ell$, or both, are selected larger than necessary, then the lower bound will underestimate the success probability.

To exemplify, the lower bound in Thm.~\ref{th:main-theorem} is derived under the assumption that one would search over offsets in the frequency~$j$ observed with the aim of finding the closest optimal frequency $j_0(z)$ for some $z \in [0, r) \inset \mathbb Z$.
If $2^{m+\ell}$ is much larger than~$r^2$, it is however not necessary to find $j_0(z)$ to solve for~$r$ in a single run:

It then suffices to find a frequency close to $j_0(z)$.
The lower bound does not account for this fact, leading it to under\-estimate the success probability if $2^{m+\ell}$ is much larger than $r^2$ --- in particular for small~$B$.
When using continued fractions-based post-processing, and for~$t$ such that $j = j_0(z) + t$, we e.g.\ need that
\begin{align}
  \frac{|\, \alpha_r \,|}{2^{m+\ell} r}
  &=
  \frac{|\, \{ rj \}_{2^{m+\ell}} \,|}{2^{m+\ell} r}
  =
  \frac{|\, \{ rj_0(z) + rt \}_{2^{m+\ell}} \,|}{2^{m+\ell} r}
  =
  \frac{|\, \{ \alpha_0(z) + rt \}_{2^{m+\ell}} \,|}{2^{m+\ell} r} \notag \\
  &=
  \frac{|\, \alpha_0(z) + rt \,|}{2^{m+\ell} r}
  =
  \frac{|\, \delta_z + t \,|}{2^{m+\ell}}
  <
  \frac{1}{2r^2}
  \quad \Rightarrow \quad
  2 \cdot |\, \delta_z + t \,| < \frac{2^{m+\ell}}{r^2}
  \label{eq:bound-t-rounding}
\end{align}
to be able to solve for $z/r$, where $\delta_z \in (-1/2, 1/2]$, see Sect.~\ref{sec:understanding-the-probability-distribution} and Sect.~\ref{sec:find-order}.
If~$2^{m+\ell}$ is much greater than $r^2$ this requirement is met for many $t \in [-B, B] \inset \mathbb Z$.
For all steps in~\refeq{eq:bound-t-rounding} to hold, it is also required that $B \in [1, B_{\max}) \inset \mathbb Z$, see~\refeq{eq:requirement-met-by-Bmax} in Sect.~\ref{sec:overview-our-contribution}.

Note furthermore that the lower bound in Lem.~\ref{lemma:probability-cm-smooth-d-uniform-z} stems from a worst-case analysis:
In practice, for most orders, the success probability is much higher than the bound indicates, and hence much higher than what the corresponding part of the lower bound in Thm.~\ref{th:main-theorem} --- and its various corollaries --- indicates.
A better lower bound can be achieved by imposing constraints on~$r$:
For instance by assuming~$r$ to be random, as random integers are unlikely to be very smooth, or by using side information on~$r$ if such information is available.

Finally, note that in the unlikely event that $d = \gcd(r, z)$ has a factor $q > cm$ so that we fail to recover~$r$, we still recover a limited set of candidates~$\tilde r_i$ for $\tilde r = r / d$.
This set includes the correct~$\tilde r$, but it is not trivial to identify~$\tilde r$ within the set.

Even so, the set may be useful:
For instance, we may use classical order-finding algorithms to attempt to find~$d$ as it is the order of $g^{\tilde r}$ --- the idea being to try to solve using all~$\tilde r_i$ in the set.\footnote{This may be combined with the techniques in this work for removing $cm$-smooth factors of~$d$.}
To take another example, when factoring an integer~$N$ via order finding, some --- and sometimes even all --- factors of~$N$ may be recovered using~\cite{completely} even if only $\tilde r = r / d$ is know as a consequence of~$d$ having a factor $q > cm$.
This follows from the analysis in~\cite{completely}, where it is stated that for
\begin{align*}
  N = \prod_{i = 1}^{n} p_i^{e_i},
\end{align*}
where $n \ge 2$, the~$p_i$ are pairwise distinct odd prime factors, and the~$e_i$ are positive integer exponents, we can afford to miss out on prime factors that divide~$p_i - 1$ for one~$p_i$ when attempting to guess a (sub-)multiple of $\lambda'(N) = \lcm(p_1 - 1, \, \ldots, \, p_n - 1)$.

If we miss out on several (large) prime factors that are associated with different prime factors~$p_i$, then these~$p_i$ may not be split apart --- but any other primes~$p_i$ will still be split apart.
Hence, we may still recover the partial or full factorization of~$N$ --- the idea again being to try to solve using all~$\tilde r_i$ in the set.
Factoring directly via~$\tilde r$ is beyond the scope of Cor.~\ref{cor:main-theorem-factor}, however, as it would complicate the corollary.

\subsection{Verifying the correctness of the result}
In this work, we give lower bounds on the success probability of recovering any order~$r$, and of completely factoring any integer~$N$, in a single run of the quantum part of Shor's order-finding algorithm.
That is to say, we give a guarantee that the correct result is returned with at least the probability indicated by the bound.
Otherwise, either no result, or a plausible yet incorrect result, is returned.

When factoring an integer~$N$, in the sense of computing a non-trivial factor~$f$ of~$N$, it is easy to verify the result:
Verify that $f \in (1, N) \inset \mathbb Z$ and that $f \mid N$.
Similarly, when computing a positive integer multiple~$r'$ of the order~$r$ of~$g$, it is easy to verify the result:
Verify that $r' \in \mathbb Z_{> 0}$ and that $g^{r'} = 1$.

When computing the order~$r$ of~$g$, or the complete factorization of~$N$‚ as in this work, verifying the result is a bit more involved, however, as explained below:

\subsubsection{Verifying that $r$ is the order of~$g$}
\label{sec:verify-computation-order-r}
To verify that~$r \in \mathbb Z_{> 0}$ such that $g^r = 1$ is the order of~$g$, factor~$r$ completely, e.g.~by using the approach described in Sect.~\ref{sec:notes-on-factoring-via-order-finding}:
For $\{ q_1, \, \ldots, \, q_k \}$ the set of distinct primes that divide~$r$, it then suffices to verify that $g^{r / q_i} \neq 1$ for all $i \in [1, k] \inset \mathbb Z$.

Two order-finding runs would then typically be required for large random~$r$; a single run in $\langle g \rangle$ to compute~$r$, and a single run in a cyclic subgroup of~$\mathbb Z_r^*$ to verify the computation by completely factoring~$r$.
Note that the completeness of the factorization must be verified, by first verifying that  $\{ q_1, \, \ldots, \, q_k \}$ is a set of pairwise distinct integers, and that $r = q_1^{e_1} \cdot \ldots \cdot q_k^{e_k}$ for $\{ e_1, \, \ldots, \, e_k \}$ a set of positive integer exponents, and by then verifying that $q_i$ is prime for all $i \in [1, k] \inset \mathbb Z$.

\subsubsection{Verifying that $q_i$ is a prime}
\label{sec:verify-computation-prime-qi}
Classical options for proving the primality of $q_i$ include using the ECPP (see  \cite{ecpp} for an overview) or AKS~\cite{aks} tests, where the former test is amongst the most efficient in practice~\cite{fkmw04, morain07} whilst the latter has a proven polynomial runtime.
If a proof of primality is not required, but only an extremely low probability of~$q_i$ not being prime, a classical probabilistic test such as Miller-Rabin~\cite{miller, rabin} may be used.

Quantum options include combining some variation of the Pocklington-Lehmer test~\cite{pocklington, lehmer} (see also Lucas~\cite[p.~441, §240]{lucas}) with a quantum factoring algorithm so as to compute the complete or partial factorizations required by the test.
This was originally proposed by Chau and Lo~\cite{cl97} with respect to using Shor's original factoring algorithm~\cite{shor94, shor97} and a variation of the test due to Brillhart et al.~\cite{bls75}.

By instead using the approach to factoring in Sect.~\ref{sec:notes-on-factoring-via-order-finding}, the test can be made more efficient:
Each complete or partial factorization required by the test may then be computed in a single quantum order-finding run with high probability of success.

\vspace{2mm}

\begin{notes}
  In practice when testing primality, as when factoring, the size of the integer and other side information typically determines the choice of algorithm.
\end{notes}

\subsubsection{Verifying that $\{ q_1, \, \ldots, \, q_n \}$ is the complete factorization of $N$}
In analogy with the procedure in Sect.~\ref{sec:verify-computation-order-r}, to verify that a set $\{ q_1, \, \ldots, \, q_n \}$ of pairwise distinct integers is the complete factorization of an integer~$N$, first verify that $N = q_1^{e_1} \cdot \ldots \cdot q_n^{e_n}$ for $\{ e_1, \, \ldots, \, e_n \}$ a set of positive integer exponents, and then verify that $q_i$ is prime for all $i \in [1, n] \inset \mathbb Z$ as described in Sect.~\ref{sec:verify-computation-prime-qi}.

\section{Summary and conclusion}
\label{sec:summary-conclusion}
We have derived a lower bound on the probability of successfully recovering the order~$r$ in a single run of the quantum part of Shor's order-finding algorithm.

The bound implies that by performing two limited searches in the classical post-processing part of the algorithm, a high success probability can be guaranteed, for any~$r$, without re-running the quantum part or increasing the exponent length when comparing to the exponent length used by Shor.
By using lattice-based post-processing and enumerating the lattice, the exponent length may in fact be slightly reduced.
The classical post-processing is efficient when accounting for searching.

In practice, the success probability is usually higher than the bound indicates, as our bound stems from a worst case analysis that holds for any~$r$ given only an upper bound~$m$ on the bit length of~$r$.
A better lower bound on the success probability can be achieved by assuming the bit length of~$r$ to be known, and by imposing constraints on~$r$:
For instance, by assuming~$r$ to be random.

Asymptotically, in the limit as~$r$ tends to infinity, the probability of successfully recovering~$r$ tends to one.
This when parameterizing the algorithm so that both the classical and quantum parts execute in polynomial time.
Already for moderate~$r$, a high success probability exceeding e.g.~$1 - 10^{-4}$ can be guaranteed.
This supports our previous results in~\cite[App.~A]{general} that were derived by means of simulations.

As a corollary, we have used our result and the classical post-processing in~\cite{completely} to derive a lower bound on the probability of factoring any integer~$N$ completely in a single order-finding run.
Again, the bound shows that a high success probability can be achieved even for moderate~$N$.
Asymptotically, in the limit as~$N$ tends to infinity, the probability of completely factoring~$N$ in a single run tends to one.
This when parameterizing the algorithm so that all parts execute in polynomial time.

\section*{Acknowledgments}
I am grateful to Johan Håstad for valuable comments and advice.
I thank Andreas Minne for proofreading early versions of this manuscript.
I thank Dennis Willsch for pointing out an issue in Alg.~4 in the first pre-print of this manuscript, and for useful discussions.
Funding and support for this work was provided by the Swedish NCSA that is a part of the Swedish Armed Forces.

\appendix
\section{Supplementary material}
In this appendix we provide some useful supplementary bounds and analyses.

\subsection{Simpler bounds on $|\, \alpha_r\,|$ and the search space in $j$}
\label{appendix:simpler-bound-search-space-j}

We know from~\refeq{eq:hw-requirement} in Sect.~\ref{sec:find-order} that we need $|\, \alpha_r \,| / 2^{m+\ell} < 1/(2r)$, for $\alpha_r = \{rj\}_{2^{m+\ell}}$, to find the convergent $z/r$ given~$j$ by expanding $j / 2^{m+\ell}$ in a continued fraction.

In this section, we give simple bounds on the probability of observing~$j$ yielding $|\, \alpha_r\,|$ of a given bit length.
These bounds enable us to coarsely capture the probability distribution in $|\, \alpha_r \,|$.
In turn, this enables us to understand what the probability is of e.g.\ finding $j_0(z)$ when searching~$j$, $j \pm 1$, $\ldots$, $j \pm B$ for some~$B$.

\subsubsection{Supporting claims}
It is helpful to first introduce a supporting definition and claim:
\begin{appdefinition}
  For $t > 0$ an integer, let $\rho(t)$ denote the probability of observing a frequency $j \in [0, 2^{m+\ell})$ that yields an argument~$\alpha_r$ such that $\left|\, \alpha_r \,\right| \in [2^{t-1}, 2^t)$.
\end{appdefinition}

\begin{appclaim}
  \label{claim:count-alpha}
  There are at most $2^{t}$ values of $j \in [0, 2^{m+\ell})$ that yield~$\alpha_r$ such that
  \begin{align*}
    \left|\, \alpha_r \,\right| \in [2^{t-1}, 2^t).
  \end{align*}
\end{appclaim}
\begin{proof}
  Let $2^{\kappa_r}$ be the largest power of two to divide~$r$.
  By the analysis in~\cite[App.~A.3]{general}, only arguments~$\alpha_r$ that are multiples of $2^{\kappa_r}$ are admissible.
  There are $2^{\kappa_r}$ distinct~$j$ that yield each admissible argument~$\alpha_r$.

  Hence, there are either $2^t$ or no admissible arguments on the interval $[2^{t-1}, 2^t)$, depending on the size of $2^{\kappa_r}$ in relation to $2^{t}$, and so the claim follows.
\end{proof}

\subsubsection{Deriving the bounds}
We now proceed, in analogy with~\cite[Lem.~2]{ekera-pp}, to derive the bounds:

\begin{applemma}
  \label{lemma:bound-rho-1}
  For $r \in [2, 2^m)$, it holds that $\rho(t) \le 2^{m - t}$.
\end{applemma}
\begin{proof}
For $0 \neq \left|\, \theta_r \,\right| \le \pi$, it holds for any $M > 0$ that
\begin{align*}
  \left|\, \sum_{b \, = \, 0}^{M - 1} \e^{\imag \theta_r b} \,\right|^2
  =
  \frac{1 - \cos(\theta_r M)}{1 - \cos(\theta_r)}
  \le
  \frac{2}{1 - \cos(\theta_r)}
  \le
  \frac{\pi^2}{\theta_r^2},
\end{align*}
where we have used Claim~\ref{claim:cos} in the last step.

Hence, it holds for $M \in \{L, L + 1\}$ and $|\, \alpha_r \,| \in [2^{t-1}, 2^t)$ that
\begin{align*}
  \refeq{eq:pr} \le
  \frac{r}{2^{2(m+\ell)}} \frac{\pi^2}{\theta_r^2} =
  \frac{r}{2^{2(m+\ell)}} \frac{\pi^2}{(2 \pi \alpha_r / 2^{m+\ell})^2} =
  \frac{r}{2^2 \alpha_r^2} <
  \frac{2^m}{2^2 \cdot 2^{2(t-1)}} =
  2^{m-2t}.
\end{align*}

It follows that the probability of observing~$j$ that yields~$\alpha_r$ such that $|\, \alpha_r \,| \in [2^{t-1}, 2^t)$ is at most $2^{m-t}$, as there are at most $2^{t}$ such~$j$ by Claim~\ref{claim:count-alpha} that each occur with probability at most $2^{m-2t}$, and so the lemma follows.
\end{proof}

\begin{applemma}
  \label{lemma:bound-rho-2}
  For $r \in [2^{m-1}, 2^m)$, it holds that $\rho(t) \le 2^{t+3-m}$.
\end{applemma}
\begin{proof}
For $\theta_r \in \mathbb R$, it holds for any $M > 0$ that
\begin{align*}
  \left|\, \sum_{b \, = \, 0}^{M - 1} \e^{\imag \theta_r b} \,\right|^2
  \le
  M^2,
\end{align*}
since the left-hand side is the square norm of a sum of~$M$ unit vectors.

Furthermore, for $\delta = \beta / r \in [0, 1)$, it holds that
\begin{align*}
  L^2 < (L + 1)^2
  & =
  \left( \floor{\frac{2^{m+\ell}}{r}} + 1 \right)^2
  =
  \left( \frac{2^{m+\ell}}{r} - \delta + 1 \right)^2 \\
  &\le
  \left( 2^{\ell + 1} + 1 - \delta \right)^2
  \le
  \left( 2^{\ell + 1} + 1 \right)^2
  <
  2^{2\ell + 3}
\end{align*}
where we have used that~$\ell$ is a positive integer in the last step, and furthermore that $r \ge 2^{m-1}$.
Hence, it holds for $M \in \{L, L + 1\}$ that
\begin{align*}
  \refeq{eq:pr} \le
  \frac{r}{2^{2(m+\ell)}} (L + 1)^2 <
  \frac{r}{2^{2(m+\ell)}} 2^{2 \ell + 3} < 2^{3-m}.
\end{align*}

It follows that the probability of observing~$j$ that yields~$\alpha_r$ such that $|\, \alpha_r \,| \in [2^{t-1}, 2^t)$ is at most $2^{t+3-m}$, as there are at most $2^{t}$ such~$j$ by Claim~\ref{claim:count-alpha} that each occur with probability at most $2^{3-m}$, and so the lemma follows.
\end{proof}

\begin{applemma}
  \label{lemma:bound-rho}
  For $r \in [2^{m-1}, 2^m)$, it holds that $\rho(t) \le \min(2^{m-t},\, 2^{t+3-m})$.
\end{applemma}
\begin{proof}
  The lemma follows by combining Lem.~\ref{lemma:bound-rho-1} and Lem.~\ref{lemma:bound-rho-2}.
\end{proof}

It follows from~\ref{lemma:bound-rho} that we are likely to observe~$j$ yielding an argument~$\alpha_r$ such that $|\, \alpha_r \,|$ is approximately of length~$m$ bits, for~$m$ the bit length of~$r$.

Hence, we can find the optimal~$j$ by searching small offsets:
If we offset~$j$ by some small integer~$t$, we offset $|\, \alpha_r \,|$ by $rt$, assuming~$\ell$ to be sufficiently large so that no modular reductions occur.
We do not need to try a very large offset to find the optimal $j = j_0(z)$, yielding an argument $\alpha_0(z) \in (-r/2, r/2] \inset \mathbb Z$.

\subsection{Order finding in Shor's factoring algorithm}
\label{appendix:notes-shor-order-finding-factoring}
Shor's original motivation for introducing the order-finding algorithm in~\cite{shor94, shor97} was to use it to factor integers, via a randomized reduction to order finding.
When Shor describes how to select parameters for the order-finding algorithm, he therefore does so from a factoring context, using the integer~$N$ to be factored as a baseline.

Specifically, for~$N$ an odd integer with at least two distinct prime factors, Shor selects~$g$ uniformly at random from~$\mathbb Z_N^*$ and quantumly computes the order~$r$ of~$g$.
Shor then uses~$r$ to attempt to split~$N$ into a product of two non-trivial factors.

For $m + \ell$ the exponent length, Shor requires that $N^2 \le 2^{m+\ell} < 2 N^2$ \cite[p.~1498]{shor97}, so that $2^{m+\ell} \ge N^2 > r^2$.
We analogously require that $2^{m+\ell} > r^2$ when solving for~$r$ in a single run using continued fractions-based or lattice-based post-processing.
To ensure that $2^{m+\ell} > r^2$ in the context of factoring $N$, we do however use $N/2$ rather than $N$ as an upper bound on $r$, see Claim~\ref{claim:carmichael} --- i.e.\ we set $m+\ell$ to the bit length\footnote{Note that since~$N$ is an odd composite, Shor's requirement that $N^2 \le 2^{m+\ell} < 2 N^2$ simply sets $m+\ell$ to the bit length of $N^2$. We instead set $m+\ell$ to the bit length of $(N/2)^2$.} of $(N/2)^2$ rather than to the bit length of $N^2$.

Hence, we have $2^{m+\ell} \ge (N/2)^2 > r^2$, or equivalently $2^{m+\ell+2} \ge N^2$, so our exponent is two bits shorter than Shor's exponent.
If we would instead have used~$N^2$ as an upper bound, our exponent would have been of the same length as Shor's.

It follows that we do not achieve an increased success probability in our analysis at the expense of increasing the exponent length $m+\ell$.
Note also that we may slightly further reduce $m+\ell$ by using lattice-based post-processing, see App.~\ref{appendix:lattice-based-post-processing}.

\subsubsection{Supporting claims}
\label{appendix:support-claims-carmichael}

\carmichaelclaim*
\begin{proof}
  By the unique factorization theorem
  \begin{align*}
    N = \prod_{i \, = \, 1}^n p_i^{e_i}
    \quad \Rightarrow \quad
    \lambda(N) = \lcm((p_1 - 1) p_1^{e_1 - 1}, \, \ldots, \, (p_n - 1) p_n^{e_n - 1})
  \end{align*}
  for~$p_i$ pairwise distinct odd primes, and~$e_i$ some positive integers, for $i \in [1, n] \inset \mathbb Z$.

  Since $p_i - 1$ is even for all $i \in [1, n] \inset \mathbb Z$, it follows that
  \begin{align*}
    \lambda(N) &= \lcm((p_1 - 1) p_1^{e_1 - 1}, \, \ldots, \, (p_n - 1) p_n^{e_n - 1}) \\
               &\le 2 \prod_{i\, = \,1}^n \frac{(p_i - 1)}{2} p_i^{e_i - 1}
                < 2^{1-n} \prod_{i\, = \,1}^n p_i^{e_i} = 2^{1-n} N
  \end{align*}
  and so the claim follows.
\end{proof}

Note that the idea of using~$N/2$ to bound~$r$ from above when factoring via order finding was seemingly first introduced by Gerjuoy~\cite{gerjuoy}, see Sect.~\ref{sec:earlier-works}.
A similar more general lemma was also later used by Bourdon and Williams~\cite[Lem.~3]{bourdon}.

\section{Continued fractions-based post-processing}
\label{appendix:continued-fractions-based-post-processing}
In this appendix, we show that when expanding $j_0(z) / 2^{m+\ell}$ in a continued fraction, the last convergent~$p/q$ with denominator~$q < 2^{(m+\ell)/2}$ is equal to~$z/r$ if $2^{m+\ell} > r^2$.

This statement implies that when solving not only~$j$ but also $j \pm 1, \, \ldots, \, j \pm B$ for $\tilde r = r / d$ where $d = \gcd(r, z)$, in the hope of thus solving~$j_0(z)$ for~$\tilde r$, it suffices to consider a single candidate for~$\tilde r$ for each offset in~$j$ considered.

Note that this statement is analogous to Shor's original statement~\cite[p.~1500]{shor97} that when $2^{m+\ell} \ge N^2 > r^2$ the last convergent~$p/q$ with denominator~$q < N$ in the expansion of $j_0(z) / 2^{m+\ell}$ must be equal to~$z/r$.
(We have merely generalized it to account for us selecting $m+\ell$ so that $2^{m+\ell} > r^2$.\footnote{In the context of factoring $N$, we use that $r < N/2$, see Claim~\ref{claim:carmichael}, and select the least $m+\ell$ such that $2^{m+\ell} \ge (N/2)^2 > r^2$. This implies that we cannot pick the last convergent with $q < N$. Instead, we pick the last convergent with $q < N/2$, or more generally with $q < 2^{(m+\ell)/2}$.})
It is furthermore analogous to the statement in Lem.~\ref{lemma:lattice-recover-tilde-r-shortest} in App.~\ref{appendix:lattice-based-post-processing} in the context of lattice-based processing.

\subsection{Preliminaries}
Before proceeding, we first need to introduce two standard supporting claims:

\begin{appclaim}
  \label{claim:convergent-existence}
  For $x \in \mathbb R$, $p \in \mathbb Z$ and $q \in \mathbb Z_{\ge 1}$, the convergent $p/q$ is in the continued fraction expansion of~$x$ if
  \begin{align*}
    \left|\, x - \frac{p}{q} \,\right| < \frac{1}{2q^2}.
  \end{align*}
\end{appclaim}
\begin{proof}
  See \cite[Thm.~184 on p.~153]{hw} for the proof.
\end{proof}

\begin{appclaim}
  \label{claim:convergent-uniqueness}
  For~$x \in \mathbb R$ and $L \in \mathbb Z_{> 1}$, there is at most one convergent~$p/q$ with~$p, q$ coprime integers such that $q \in (0, L)$, $p \in (0, q)$ and
  \begin{align*}
    \left|\, x - \frac{p}{q} \,\right| \le \frac{1}{2 L^2}.
  \end{align*}
\end{appclaim}
\begin{proof}
  Suppose the contrary that there is a second convergent $p'/q' \neq p/q$ with~$p', q'$ coprime integers such that $q' \in (0, L)$, $p' \in (0, q')$ and
  \begin{align*}
    \left|\, x - \frac{p'}{q'} \,\right| \le \frac{1}{2 L^2}.
  \end{align*}

  By the triangle inequality, it must then be that
  \begin{align}
    \label{eq:claim-convergent-uniqueness-eq1}
    \left|\, \frac{p}{q} - \frac{p'}{q'} \,\right|
    =
    \left|\, \left( x - \frac{p'}{q'} \right) - \left( x - \frac{p}{q} \right) \,\right|
    \le
    \left|\, x - \frac{p'}{q'} \,\right| + \left|\, x - \frac{p}{q} \,\right|
    \le
    \frac{1}{L^2}.
  \end{align}
  At the same time
  \begin{align}
    \label{eq:claim-convergent-uniqueness-eq2}
    \left|\, \frac{p}{q} - \frac{p'}{q'} \,\right|
    =
    \left|\, \frac{p \cdot q'}{q \cdot q'} - \frac{p' \cdot q}{q' \cdot q} \,\right|
    =
    \frac{|\, p \cdot q' - p' \cdot q \,|}{q' \cdot q}
    \ge
    \frac{1}{(L-1)^2}
    >
    \frac{1}{L^2},
  \end{align}
  as $q' \cdot q \in [1, (L-1)^2] \inset \mathbb Z$ and $0 \neq p \cdot q' - p' \cdot q \in \mathbb Z$.

  The claim follows from the contradiction between~\refeq{eq:claim-convergent-uniqueness-eq1} and~\refeq{eq:claim-convergent-uniqueness-eq2}.
\end{proof}

\subsection{Identifying the convergent $z/r$}
\continuedfractionslemma*
\begin{proof}
  It follows from Claim~\ref{claim:convergent-existence}, and the fact that
  \begin{align*}
    \left|\, \frac{j_0(z)}{2^{m+\ell}} - \frac{z}{r} \,\right|
    =
    \left|\, \frac{rj_0(z)}{2^{m+\ell} r} - \frac{2^{m+\ell} z}{2^{m+\ell} r} \,\right|
    =
    \frac{\left|\, \{r j_0(z)\}_{2^{m+\ell}} \,\right|}{2^{m+\ell} r}
    \le
    \frac{1}{2 \cdot 2^{m+\ell}}
    <
    \frac{1}{2r^2},
  \end{align*}
  that the convergent $p/q = z/r$ must occur in the continued fraction expansion of $x = j_0(z) / 2^{m+\ell}$, where we note explicitly that $q \le r < 2^{(m+\ell)/2}$ as $2^{m+\ell} > r^2$.

  Trivially $j_0(z) = 0$ if and only if $z = 0$, in which case $p/q = 0/1 = z/r$ is the only convergent in the continued fraction expansion of $x = 0$.
  Suppose that $z \neq 0$:

  By Claim~\ref{claim:convergent-uniqueness}, there is then at most one convergent $p/q$ such that
  \begin{align*}
    \left|\, \frac{j_0(z)}{2^{m+\ell}} - \frac{p}{q} \,\right| \le \frac{1}{2 \cdot 2^{m+\ell}}
  \end{align*}
  with $p, q$ coprime integers such that $q \in (0, 2^{(m+\ell)/2})$ and $p \in (0, q)$.
  This convergent $p/q$ must be equal to $z/r$, with $p = z/d$ and $q = r/d = \tilde r$ for $d = \gcd(r, z)$.

  Hence, as successive convergents $p/q$ in the expansion of $x = j_0(z) / 2^{m+\ell}$ yield increasingly good approximations to~$x$ and therefore must have strictly increasing denominators, we will recover $z/r$ if we pick the last convergent in the expansion with denominator $q < 2^{(m+\ell)/2}$, and so the lemma follows.
\end{proof}

\section{Lattice-based post-processing}
\label{appendix:lattice-based-post-processing}
In this appendix, we use lattice-based post-processing to recover $\tilde r = r / \gcd(r, z)$ from an optimal frequency $j = j_0(z)$ for any $z \in [0, r) \inset \mathbb Z$ in the setting where~$m$ is selected so that $2^m > r$, and where $\ell = m - \Delta$ for some $\Delta \in [0, m) \inset \mathbb Z$.

To this end, we essentially follow~\cite{general}, except that we specifically consider and analyze the two-dimensional case, and that we do so under the assumption that~$j$ is optimal.
We bound the number of vectors that must at most be enumerated in the lattice to guarantee that~$\tilde r$ may be recovered from one of the vectors enumerated.

\subsection{Earlier related works}
Before proceeding, let us review some earlier works on post-processing the output from Shor's order-finding algorithm, and their respective relations to this work:

\subsubsection{Notes on the relation to Seifert's work}
As stated in Sect.~\ref{sec:earlier-works}, Seifert~\cite{seifert} explores tradeoffs by letting $\ell \sim m/s$ for some integer~$s > 1$.
Each run then provides at least\footnote{If~$m$ is the bit length of~$r$, the algorithm yields $\sim \ell$ bits of information on~$r$.
If~$m$ is greater than the bit length of~$r$, the algorithm yields more than $\sim \ell$ bits of information on~$r$.}~$\sim \ell$ bits of information on~$r$, so~$\sim s$ runs are required to ensure there is sufficient information available to solve for~$r$.

Seifert~\cite{seifert} first performs $n \ge s$ runs of the quantum order-finding algo\-rithm with $\ell = m/s$ in the hope of obtaining a set of~$n$ good frequencies $\{ j_1, \, \ldots, \, j_n \}$.
He then jointly post-processes this set of frequencies by generalizing Shor's original continued fractions-based post-processing algorithm to higher dimensions.

In~\cite[App.~A and Sect.~6.2]{general}, Ekerå instead uses lattice-based post-processing, that is adapted from~\cite{ekera-pp, ekera-hastad}, for both Shor's and Seifert's algorithms.
Furthermore, Ekerå relaxes the requirement on the frequencies by capturing the probability distribution induced by the quantum algorithm.
See~\cite[App.~A and Sect.~6.2]{general}, the literature review in Sect.~\ref{sec:earlier-works-simulations}, and App.~\ref{appendix:simpler-bound-search-space-j}, for further details.

The post-processing algorithm that we introduce in this app\-en\-dix is modeled upon~\cite{general} via~\cite{ekera-pp, ekera-hastad}, but it solves a single optimal frequency $j_0(z)$ for $\tilde r = r/d$, from which~$r$ may then be recovered via Alg.~\ref{alg:recover-r} or Alg.~\ref{alg:recover-r-tree} when $d = \gcd(r, z)$ is $cm$-smooth.
By searching offsets in the frequency observed, we find $j_0(z)$ with high probability.
This search is feasible to mount when post-processing a single frequency, as the search space is then small.

\subsubsection{Notes on the relation to Koenecke's and Wocjan's work}
Koenecke and Wocjan~\cite{koenecke-wocjan} observe that the problem of finding the convergent $z/r$ in the continued fraction expansion in Shor's algorithm may be perceived as a lattice problem, and be solved using a slightly different lattice-based post-processing:

Specifically, they seek to recover~$r$ from two optimal frequencies $j_1 = j_0(z_1)$ and $j_2 = j_0(z_2)$, such that $z_1$ and $z_2$ are coprime, returned from two separate runs of the quantum part of Shor's original order-finding algorithm.
They hence require at least two runs of the quantum part, with an exponent of length as in Shor's original algorithm, so $\ell \sim m$.
Their lattice basis is different from that in~\cite{general}.

The post-processing algorithm in this appendix --- that stems from the post-processing in~\cite{general} via~\cite{ekera-pp, ekera-hastad} --- requires only a single optimal frequency, and hence only a single run of the quantum part, provided it yields a frequency that is close enough to an optimal frequency $j_0(z)$ for it to be found by searching.
It recovers $\tilde r = r/d$, from which~$r$ may be recovered when $d = \gcd(r, z)$ is $cm$-smooth.

\subsubsection{Notes on the relation to Knill's work}
As stated in Sect.~\ref{sec:earlier-works}, Knill~\cite{knill} explores tradeoffs between the exponent length, the search space in the classical post-processing, and the success probability, in the context of using continued fractions-based post-processing:

For $\ell = m - \Delta$ for some small~$\Delta$, Knill essentially proposes to solve an optimal frequency $j = j_0(z)$ for convergents with denominators on successive limited intervals using Lehmer's algo\-rithm~\cite[Alg.~1.3.13 on p.~22]{cohen} so as to recover~$z/r$.
This is similar to the lattice-based post-processing that we introduce in this appendix.

Knill says to run the quantum part of Shor's order-finding algorithm twice for the same~$g$ to obtain two optimal frequencies $j_1 = j_0(z_1)$ and $j_2 = j_0(z_2)$, to post-process these independently to obtain $z_1 / r$ and $z_2 / r$, and to then take the least common multiple of the denominators of the convergents as the candidate for~$r$.

The whole process then requires at least two runs.
As stated above, our post-processing algorithm requires only a single optimal frequency, and hence only a single run of the quantum part, provided it yields a frequency that is close enough to an optimal frequency $j_0(z)$ for it to be found by searching.
It recovers~$\tilde r = r/d$, from which~$r$ may be recovered when $d = \gcd(r, z)$ is $cm$-smooth.

\subsection{Preliminaries}
We follow~\cite{general}, and let~$\mathcal L$ be the lattice spanned by $\vec b_1 = (j, 1/2)$ and $\vec b_2 = (2^{m+\ell}, 0)$, where $j = j_0(z)$ is an optimal frequency for some peak index $z \in [0, r) \inset \mathbb Z$.

Note that, compared to~\cite{general}, we have scaled the second component of $\vec b_1$ slightly by a factor of $1/2$ as $|\, \alpha_r \,| = |\, \alpha_0(z) \,| \le r/2$ when~$j = j_0(z)$.
This yields slightly better constants in the analysis:
In particular, for $\Delta = 0$ it makes the lattice-based post-processing perform on par with continued fractions-based post-processing.

All vectors in~$\mathcal L$ may be written on the form
\begin{align*}
  \vec v(m_1, m_2)
  =
  m_1 \vec b_1 - m_2 \vec b_2
  =
  (m_1 j - 2^{m+\ell} m_2, m_1 / 2) \in \mathcal L
\end{align*}
for $m_1, m_2 \in \mathbb Z$.
In particular, for $d = \gcd(r, z)$, the vector
\begin{align*}
  \vec u = \vec v(r/d, z/d) = (rj - 2^{m+\ell} z, r / 2) / d = (\alpha_0(z), r / 2) / d = (\alpha_0(z) / d, \tilde r / 2) \in \mathcal L,
\end{align*}
and it has $\tilde r = r/d$ as its second component.

Furthermore, as $d \ge 1$, $r < 2^m$ and $|\, \alpha_r \,| = |\, \alpha_0(z) \,| \le r/2$, we have that
\begin{align}
  |\, \vec u \,| = \frac{\sqrt{\alpha_0(z)^2 + (r/2)^2}}{d^2} \le \sqrt{\left( \frac{r}{2} \right)^2 + \left( \frac{r}{2} \right)^2} = \frac{r}{\sqrt{2}} < 2^{m - \frac{1}{2}}, \label{eq:lattice-bound-u}
\end{align}
where $|\, \vec x \,|$ denotes the Euclidean norm of $\vec x \in \mathcal L$, both above and in what follows.

The idea is now to enumerate all vectors in~$\mathcal L$ that are within a circle of radius $2^{m-\frac{1}{2}}$ centered at the origin to find~$\vec u$ and hence~$\tilde r$.
Alg.~\ref{alg:filter-tilde-r} may be used to filter the candidates for~$\tilde r$ thus generated when~$d$ is $cm$-smooth.
A better option is to use the optimized filtering algorithm in App.~\ref{appendix:lattice-based-post-processing-filtering-candidates-efficiently} that leverages the fact that all candidates for~$\tilde r$ stem from vectors that are in~$\mathcal L$.
Once~$\tilde r$ has been found, the order~$r$ may be recovered from~$\tilde r$ using Alg.~\ref{alg:recover-r} or~\ref{alg:recover-r-tree} when~$d$ is $cm$-smooth, see Sect.~\ref{sec:find-order-recover-r-from-tilde-r}.

Note that there are other possible approaches:
We could e.g.\ accept a larger $|\, \alpha_r \,|$, as in~\cite{general}, to avoid first searching for the optimal frequency $j_0(z)$ close to the frequency~$j$ observed, and then solving at most all of these frequencies using lattice-based techniques.
This at the expense of enumerating at most all vectors within a much larger circle in~$\mathcal L$ --- but at the benefit of only performing a single such large enumeration, as opposed to many small enumerations.

Our objective in this appendix is to keep the post-processing simple to analyze, and to align it with the analysis in the main part of the paper.
Therefore, we take the three-step approach of first searching for $j_0(z)$ given some close frequency, then recovering~$\tilde r$ from $j = j_0(z)$, and finally recovering~$r$ from~$\tilde r$.

\subsection{Notation and supporting claims}
Up to signs, let~$\vec s_1$ of norm $\lambda_1$ be a shortest non-zero vector in~$\mathcal L$, and let~$\vec s_2$ of norm $\lambda_2 \ge \lambda_1$ be the shortest non-zero vector in~$\mathcal L$ that is linearly independent to~$\vec s_1$.
Note that~$\vec s_1$ and~$\vec s_2$ may be found efficiently using Lagrange's algorithm~\cite{lagrange, nguyen}.

Furthermore, let $\vec s^{\perp}_2$ of norm $\lambda_2^{\perp}$, and $\vec s^{\parallel}_2$ of norm $\lambda_2^{\parallel}$, be the components of~$\vec s_2$ that are orthogonal and parallel to~$\vec s_1$, respectively.

\begin{appclaim}
  \label{claim:lattice-det}
  It holds that $\lambda_1 \lambda_2^{\perp} = 2^{m+\ell-1}$.
\end{appclaim}
\begin{proof}
  The claim follows from the fact that $\lambda_1 \lambda_2^{\perp} = \det \mathcal L = 2^{m+\ell-1}$ is the area of the fundamental parallelogram in~$\mathcal L$.
\end{proof}

\begin{appclaim}
  \label{claim:lattice-ineq}
  It holds that $\lambda_2^{\parallel} \le \lambda_1 / 2$ and as a consequence that $\lambda_2^{\perp} \ge \sqrt{3} \, \lambda_2 / 2$.
\end{appclaim}
\begin{proof}
  For $\text{proj}_{\vec s_1} (\vec s_2)$ the projection of~$\vec s_2$ onto~$\vec s_1$, we have that
  \begin{align*}
    \mu &= \text{proj}_{\vec s_1} (\vec s_2) = \frac{\langle \vec s_1, \vec s_2 \rangle}{|\, \vec s_1 \,|^2},
    &
    \vec s_2^{\parallel} &= \mu \vec s_1,
    &
    \vec s_2^{\perp} &= \vec s_2 - \vec s_2^{\parallel} = \vec s_2 - \mu \vec s_1.
  \end{align*}

  It must be that $|\, \mu \,| \le 1/2$.
  Otherwise~$\vec s_1$ and $\vec s'_2 = \vec s_2 - \round{\mu} \cdot \vec s_1$ form a basis for~$\mathcal L$, with $|\, \vec s'_2 \,| < |\, \vec s_2 \,|$.
  This is inconsistent with~$\vec s_1$ and~$\vec s_2$ forming a reduced basis up to sign.
  It follows that $\lambda_2^{\parallel} = |\, \vec s_2^{\parallel} \,| = |\, \mu \vec s_1 \,| \le \lambda_1 / 2$.
  Furthermore
  \begin{align*}
    \lambda_2^2
    =
    (\lambda_2^{\perp})^2 + (\lambda_2^{\parallel})^2
    \leq
    (\lambda_2^{\perp})^2 + \lambda_1^2/4
    \leq
    (\lambda_2^{\perp})^2 + \lambda_2^2/4
    \quad \Rightarrow \quad
    (\lambda_2^{\perp})^2
    \ge
    3 \lambda_2^2 / 4
  \end{align*}
  which implies $\lambda_2^{\perp} \ge \sqrt{3} \, \lambda_2 / 2$ by taking the root, and so the claim follows.
\end{proof}

\begin{appclaim}
  \label{claim:size-lambda-two-perp-u}
  Suppose that $2^{m+\ell - 1} > |\, \vec u \,|^2$.
  Then $\lambda_2^{\perp} > |\, \vec u \,|$.
\end{appclaim}
\begin{proof}
  As $\lambda_1 \leq |\, \vec u \,|$, we have $\lambda_2^{\perp} = 2^{m+\ell-1}/\lambda_1 \geq 2^{m+\ell-1} / |\, \vec u \,| > |\, \vec u \,|$ where we have used Claim~\ref{claim:lattice-det}, and the supposition in the claim in the last step.
\end{proof}

\subsection{Bounding the complexity of the enumeration}
If~$\ell$ is sufficiently large as a function of~$r$ and~$m$, then we can immediately recover~$\tilde r$ by reducing the basis $(\vec b_1, \vec b_2)^{\text{T}}$ for~$\mathcal L$ to $(\vec s_1, \vec s_2)^{\text{T}}$ with Lagrange's algorithm~\cite{lagrange, nguyen}:

\latticelemmashortest*
\begin{proof}
  As $|\, \vec u \,| \le r / \sqrt{2}$ by~\refeq{eq:lattice-bound-u}, the supposition in the lemma that $2^{m+\ell} > r^2$ implies that $2^{m+\ell-1} > |\, \vec u \,|^2$, and by Claim~\ref{claim:size-lambda-two-perp-u}, $\lambda_2^{\perp} > |\, \vec u \,|$ if $2^{m+\ell-1} > |\, \vec u \,|^2$.

  Hence, $\lambda_2^{\perp} > |\, \vec u \,|$, so it must be that $\vec u$ is a multiple of~$\vec s_1$.
  In fact, it must be that $\vec u = (rj - 2^{m+\ell} z, r / 2) / \gcd(r, z)$ is equal to~$\vec s_1$ up to sign, as the two components of $\vec u$ are coprime when scaled up by a factor of two, and so the lemma follows.
\end{proof}

By Lem.~\ref{lemma:lattice-recover-tilde-r-shortest}, we can immediately recover~$\tilde r$ by reducing the basis for the lattice provided that $2^{m+\ell} > r^2$.
This is analogous to the situation that arises when solving using continued fractions-based post-processing.
As $2^m > r$, it suffices to pick $\ell \ge m$ to meet the requirement.
If~$\ell$ is less than~$m$ --- say that $\ell = m - \Delta$ for some $\Delta \in [0, m) \inset \mathbb Z$ --- then we can still find~$\tilde r$ by enumerating at most $\deltaexpr$ vectors in~$\mathcal L$ that lie within a ball of a radius $2^{m - \frac{1}{2}}$ centered at the origin:

\latticelemma*
\begin{proof}
There are two cases that we treat separately:
\begin{enumerate}
  \item Suppose $\lambda_2^{\perp} \ge 2^{m - \frac{1}{2}}$:
  Then $\lambda_2^{\perp} > |\, \vec u \,|$ by~\refeq{eq:lattice-bound-u}, so $\vec u = (rj - 2^{m+\ell} z, r / 2) / d$ must be a multiple of~$\vec s_1$.
  In fact, $\vec u$ is equal to~$\vec s_1$ up to sign, as the two components of~$\vec u$ are coprime when scaled up by a factor of two, so we find $\tilde r / 2 = r / (2d)$ and hence~$\tilde r$ up to sign in the second component of~$\vec s_1$.

  In this case there is hence no need to enumerate~$\mathcal L$:
  For $\vec s_1 = (s_{1,1}, s_{1,2})$, it suffices to include $2 \cdot |\, s_{1,2} \,| = \tilde r$ in the set of candidates.

  \item Suppose $\lambda_2^{\perp} < 2^{m - \frac{1}{2}}$: \label{case:lemma-lattice-recover-tilde-r-case-two-enumerate}
  In this case, we enumerate all vectors on the form
  \begin{align*}
    \vec w(m_1, m_2) = m_1 \vec s_1 + m_2 \vec s_2 = (w_1, w_2) \in \mathcal L
  \end{align*}
  for $m_1, m_2 \in \mathbb Z$ such that $|\, \vec w(m_1, m_2) \,| < 2^{m-\frac{1}{2}}$.
  Then~$\vec u$ is amongst the vectors enumerated, as $|\, \vec u \,| < 2^{m-\frac{1}{2}}$ by~\refeq{eq:lattice-bound-u}.
  For $\vec w = \vec u$, it holds that $w_2 = \tilde r / 2$, so including~$2 w_2$ in the set of candidates includes~$\tilde r$.

  As $|\, m_2 \,| < 2^{m-\frac{1}{2}} / \lambda_2^{\perp}$, we need to consider at most $1 + 2 \cdot 2^{m - \frac{1}{2}} / \lambda_2^{\perp}$ values of~$m_2$.
  For each value of~$m_2$, we need to consider at most $1 + 2 \cdot 2^{m - \frac{1}{2}} / \lambda_1$ values of~$m_1$.
  The number of vectors to enumerate is hence at most
  \begin{align}
    M
    &=
    (1 + 2^{m + \frac{1}{2}} / \lambda_1) (1 + 2^{m + \frac{1}{2}} / \lambda_2^{\perp}) \notag \\
    &<
    (1 + 2^{m + \frac{3}{2}} / (\sqrt{3} \lambda_1)) (1 + 2^{m + \frac{1}{2}} / \lambda_2^{\perp}) \label{eq:M-scale} \\
    &<
    3^2 \cdot (2^{m + \frac{1}{2}} / (\sqrt{3} \lambda_1)) (2^{m - \frac{1}{2}} / \lambda_2^{\perp}) \label{eq:M-assumption} \\
    &=
    3 \sqrt{3} \cdot 2^{2m} / (\lambda_1 \lambda_2^{\perp}) \label{eq:M-det}
    =
    6 \sqrt{3} \cdot 2^{m-\ell}
    =
    \deltaexpr
  \end{align}
  where, in~\refeq{eq:M-det}, we have used Claim~\ref{claim:lattice-det}.

  As for the inequality in step~\refeq{eq:M-assumption}, we supposed $\lambda_2^{\perp} < 2^{m - \frac{1}{2}}$, so $2^{m + \frac{1}{2}} / \lambda_2^{\perp} > 2$ which implies that $1 + 2^{m + \frac{1}{2}} / \lambda_2^{\perp} < (\frac{1}{2}+1) \cdot 2^{m + \frac{1}{2}} / \lambda_2^{\perp} = 3 \cdot 2^{m - \frac{1}{2}} / \lambda_2^{\perp}$.

  Furthermore, by Claim~\ref{claim:lattice-ineq}, we have that $\lambda_1 \le \lambda_2 \le 2 \lambda_2^{\perp} / \sqrt{3} < 2^{m + \frac{1}{2}} / \sqrt{3}$, so $2^{m + \frac{3}{2}} / (\sqrt{3} \lambda_1) > 2$, which implies that $1 + 2^{m + \frac{3}{2}} / (\sqrt{3} \lambda_1) < 3 \cdot 2^{m + \frac{1}{2}} / (\sqrt{3} \lambda_1)$.

  (This is why we scaled up the main term in the first factor in~\refeq{eq:M-scale} by $2 / \sqrt{3}$.)
\end{enumerate}

As there are no vectors to enumerate in the first case, and at most~$\deltaexpr$ vectors to enumerate in the second case, the lemma follows.
\end{proof}

\subsubsection{Notes on improving the bound on the enumeration complexity}
\label{appendix:lattice-based-post-processing-improved-bound}
Consider the enumeration in case~\ref{case:lemma-lattice-recover-tilde-r-case-two-enumerate} in the proof of Lem.~\ref{lemma:lattice-recover-tilde-r} in the previous section:

If $\vec w(m_1, m_2) = (w_1, w_2) \in \mathcal L$ is within the circle to be enumerated, then so is~$-\vec w$, but it suffices to find $\vec w$.
One way to avoid including $-\vec w$ in the enumeration is to first iterate over $m_2$ as in the proof, and to then for each $m_2$ only iterate over $m_1$ that yield non-negative $w_2$.
This essentially enumerates only the top semicircle, and hence improves the bound in Lem.~\ref{lemma:lattice-recover-tilde-r} by approximately a factor of two.

It is possible to obtain an even better bound, by e.g.~using that
\begin{align}
  |\, w_1 \,| &\le |\, \alpha_0(z) \,| \le r/2 < 2^{m-1},
  &
  0 \le w_2 &\le r/2 < 2^{m-1}, \label{eq:restrict-w}
\end{align}
to further restrict $(w_1, w_2)$, and hence $(m_1, m_2)$, when performing the enumeration.

In summary, the number of vectors in~$\mathcal L$ that need to be enumerated to find~$\vec u$ is $\ordo(2^\Delta)$.
The constant of $6 \sqrt{3}$ in Lem.~\ref{lemma:lattice-recover-tilde-r} may be slightly improved, at the expense of slightly complicating the procedure and the analysis.

\subsubsection{Notes on filtering the candidates for $\tilde r$ efficiently}
\label{appendix:lattice-based-post-processing-filtering-candidates-efficiently}
As described in Sect.~\ref{sec:solving-candidate-set-for-r}, the set of candidates for~$\tilde r$ generated by enumerating~$\mathcal L$ may be passed to Alg.~\ref{alg:recover-r} or Alg.~\ref{alg:recover-r-tree} to be solved for~$r$, or to Alg.~\ref{alg:recover-multiple-of-r} to be solved for a positive integer multiple~$r'$ of~$r$.
Prior to passing the candidates for~$\tilde r$ to any of these algorithms, it is advantageous to first filter the candidates for~$\tilde r$.
One option is to use Alg.~\ref{alg:filter-tilde-r} that performs at most one $m$-bit exponentiation per candidate.

At minimum, the two requirements in~\refeq{eq:restrict-w} should be checked before passing~$2 w_2$ as a candidate for~$\tilde r$ to Alg.~\ref{alg:filter-tilde-r}, so as to avoid performing exponentiations in step~\ref{alg:filter-tilde-r-step:test} of Alg.~\ref{alg:filter-tilde-r} for candidates that can immediately be trivially dismissed.

The amount of exponentiation work that needs to be performed may be further reduced by leveraging the structure of~$\mathcal L$:
For $\vec s_1 = (s_{1,1}, s_{1,2})$ and $\vec s_2 = (s_{2,1}, s_{2,2})$, we may pre-compute $x_{1} = x^{2 s_{1,2}}$ and $x_{2} = x^{2 s_{2,2}}$, where as in Alg.~\ref{alg:filter-tilde-r} we let
\begin{align*}
  x = g^e
  \quad \text{for} \quad
  e = \prod_{q \in \mathcal P(cm)} q^{\lfloor \log_q cm \rfloor}.
\end{align*}

For $\vec w(m_1, m_2) = m_1 \vec s_1 + m_2 \vec s_2 = (w_1, w_2)$, we then have that $x^{2 w_2} = x_{1}^{m_1} x_{2}^{m_2}$, allowing us to test if $x^{2 w_2} = x_{1}^{m_1} x_{2}^{m_2} = 1$ for each candidate $2 w_2$ for~$\tilde r$ in step~\ref{alg:filter-tilde-r-step:test} of Alg.~\ref{alg:filter-tilde-r} by raising $x_1, x_2$ to small $m_1, m_2$, instead of by raising~$x$ to~$2w_2$.

\subsubsection{Notes on efficiently solving a range of offsets in $j$ for $\tilde r$}
\label{appendix:lattice-based-post-solving-range-of-offsets-in-j-efficiently}
In practice, we do not know if the frequency~$j$ observed is optimal, so we solve not only~$j$ but also $j \pm 1, \, \ldots, \, j \pm B$ for~$\tilde r$ in the hope of thus solving $j_0(z)$ for~$\tilde r$.

To do this efficiently using lattice-based post-processing, for~$j$ the frequency observed, we may first setup the basis $(\vec b_1(j), \vec b_2)^{\text{T}}$ for the lattice $\mathcal L(j)$ where
\begin{align*}
  \vec b_1(j) = (j, 1/2)
  \quad \text{ and } \quad
  \vec b_2 = (2^{m+\ell}, 0),
\end{align*}
and Lagrange-reduce it to obtain $(\vec s_1(j), \vec s_2(j))^{\text{T}}$.
In this process, we may easily also compute and return row multiples $\nu_{1,1}(j), \, \nu_{1,2}(j), \, \nu_{2,1}(j), \, \nu_{2,2}(j) \in \mathbb Z$ such that
\begin{align*}
  \vec s_1(j) &= \nu_{1,1}(j) \cdot \vec b_1(j) + \nu_{1,2}(j) \cdot \vec b_2, \\
  \vec s_2(j) &= \nu_{2,1}(j) \cdot \vec b_1(j) + \nu_{2,2}(j) \cdot \vec b_2.
\end{align*}

When reducing the basis for $\mathcal L(j \pm k)$, for $k = 1, \, \ldots, \, B$, we may then recursively use the row multiples computed when reducing the basis for the lattice $\mathcal L(j \pm (k - 1))$ to compute the basis $(\vec b'_1(j \pm k), \, \vec b'_2(j \pm k))^{\text{T}}$ for $\mathcal L(j \pm k)$ where
\begin{align*}
  \vec b'_1(j \pm k) &= \nu_{1,1}(j \pm (k - 1)) \cdot \vec b_1(j \pm k) + \nu_{1,2}(j \pm (k - 1)) \cdot \vec b_2, \\
  \vec b'_2(j \pm k) &= \nu_{2,1}(j \pm (k - 1)) \cdot \vec b_1(j \pm k) + \nu_{2,2}(j \pm (k - 1)) \cdot \vec b_2,
\end{align*}
and then Lagrange-reduce this basis to obtain the basis $(\vec s_1(j \pm k), \, \vec s_2(j \pm k))^{\text{T}}$ and the row multiples $\nu_{1,1}(j \pm k)$, $\nu_{1,2}(j \pm k)$, $\nu_{2,1}(j \pm k)$ and $\nu_{2,2}(j \pm k)$.

This is typically much faster than independently Lagrange-reducing the bases $(\vec b_1(j \pm k), \vec b_2)^{\text{T}}$ for $\mathcal L(j \pm k)$ for $k = 1, \, \ldots, \, B$, since $(\vec b'_1(j \pm k), \, \vec b'_2(j \pm k))^{\text{T}}$ is easy to compute and heuristically likely to already be close to Lagrange-reduced.

Given the reduced bases~$(\vec s_1(j \pm k), \, \vec s_2(j \pm k))^{\text{T}}$ for~$\mathcal L(j \pm k)$ for $k = 0, \, \ldots, \, B$, we may finally proceed as previously outlined to solve $j \pm k$ for~$\tilde r$ by enumerating short vectors in $\mathcal L(j \pm k)$, or by considering only the shortest vector in $\mathcal L(j \pm k)$.

\section{Proofs of supporting lemmas and claims}
\label{appendix:proofs}
In this appendix we provide proofs for supporting lemmas and claims.

\subsection{Approximating $P(\alpha_r)$ by $\widetilde{P}(\alpha_r)$}
\label{appendix:proofs-approximating}

\mvtclaim*
\begin{proof}
  By the mean value theorem (MVT), we have for some $t \in (u, v)$ that
  \begin{align*}
    \cos'(t) = \frac{\cos(u) - \cos(v)}{u - v}
  \end{align*}
  which implies
  \begin{align*}
    |\, \cos(u) - \cos(v) \,| &=   |\, u - v \,| \cdot |\, \cos'(t) \,| \\
                              &< |\, u - v \,| \cdot \max\left( |\,u\,|, |\,v\,| \right) 
  \end{align*}
  since $|\, \cos'(t) \,| = |\, \sin(t) \,| \le |\, t \,|$, again by the MVT when $t \neq 0$ as
  \begin{align*}
    |\, \sin(t) \,| = |\, \sin(t) - \sin(0) \,| &= |\, t \,| \cdot |\, \cos(t') \,| \le |\, t \,|
  \end{align*}
  for some $t' \in (0, t)$, and trivially when $t = 0$, and so the claim follows.
\end{proof}

\steponeclaimone*
\begin{proof}
  For $\delta = \beta / r \in [0, 1)$, it holds that
  \begin{align*}
    \frac{L + 1}{2^{m+\ell}}
    &=
    \frac{1}{2^{m+\ell}} \left( \floor{\frac{2^{m+\ell}}{r}} + 1 \right)
    =
    \frac{1}{2^{m+\ell}} \left( \frac{2^{m+\ell}}{r} - \delta + 1 \right)
    =
    \frac{1}{r} + \frac{1 - \delta}{2^{m+\ell}} \\
    &<
    \frac{1}{r} + \frac{1 - \delta}{2r}
    \le
    \frac{1}{r} + \frac{1}{2r}
    =
    \frac{3}{2r},
  \end{align*}
  as $r < 2^m$ and $\ell \ge 1$, from which the claim follows.
\end{proof}

\steponeclaimtwo*
\begin{proof}
  For $\delta = \beta / r \in [0, 1)$, it holds that
  \begin{align*}
    \left|\, \frac{L+1}{2^{m+\ell}} - \frac{1}{r} \,\right|
    &=
    \left|\, \frac{1}{2^{m+\ell}} \left( \floor{\frac{2^{m+\ell}}{r}} + 1 \right) - \frac{1}{r} \,\right|
    =
    \left|\, \frac{1}{2^{m+\ell}} \left( \frac{2^{m+\ell}}{r} - \delta + 1 \right) - \frac{1}{r} \,\right| \\
    &=
    \left|\, \frac{1}{r} + \frac{1 - \delta}{2^{m+\ell}} - \frac{1}{r} \,\right|
    =
    \left|\, \frac{1 - \delta}{2^{m+\ell}} \,\right|
    \le
    \frac{1}{2^{m+\ell}},
  \end{align*}
  and analogously that
  \begin{align*}
    \left|\, \frac{L}{2^{m+\ell}} - \frac{1}{r} \,\right|
    &=
    \left|\, \frac{1}{2^{m+\ell}} \floor{\frac{2^{m+\ell}}{r}} - \frac{1}{r} \,\right|
    =
    \left|\, \frac{1}{2^{m+\ell}} \left( \frac{2^{m+\ell}}{r} - \delta \right) - \frac{1}{r} \,\right| \\
    &=
    \left|\, \frac{1}{r} - \frac{\delta}{2^{m+\ell}} - \frac{1}{r} \,\right|
    =
    \left|\, \frac{\delta}{2^{m+\ell}} \,\right|
    <
    \frac{1}{2^{m+\ell}},
  \end{align*}
  and so the claim follows.
\end{proof}

\steponeclaimcos*
\begin{proof}
  For $\phi = 0$, the claim trivially holds.
  Given the series expansion
  \begin{align*}
    \label{eq:cos-proof-inequality}
    \cos \phi
    =
    \sum_{k \, = \, 0}^\infty
      \frac{(-1)^k \phi^{2k}}{(2k)!}
    = 1 - \frac{\phi^2}{2!} + \frac{\phi^4}{4!} - \frac{\phi^6}{6!} + \ordo(\phi^8)
  \end{align*}
  we have for $0 < |\, \phi \,| \le \pi$ that
  \begin{align*}
    \frac{\phi^2}{2} - \frac{\phi^4}{4!} \le 1 - \cos \phi \le \frac{\phi^2}{2}
  \end{align*}
  as the series is alternating, and as it holds for any $k \ge 1$ that
  \begin{align*}
    \frac{\phi^{2(k+1)}}{(2(k+1))!} < \frac{\phi^{2k}}{(2k)!}.
  \end{align*}
  This proves the upper bound on $1 - \cos \phi$.

  As for the lower bound, it holds that
  \begin{align}
    \frac{2 \phi^2}{\pi^2}
    \le
    \frac{\phi^2}{2} - \frac{\phi^4}{4!}
    \le
    1 - \cos \phi
  \end{align}
  for $0 < |\, \phi \,| \le \varphi = \frac{2}{\pi} \sqrt{3(\pi^2 - 4)}$, as
  \begin{align*}
    \frac{2 \phi^2}{\pi^2}
    \le
    \frac{\phi^2}{2} - \frac{\phi^4}{4!}
    \quad\Rightarrow\quad
    \phi^2
    \le
    4! \left( \frac{1}{2} - \frac{2}{\pi^2} \right)
    =
    \frac{12}{\pi^2} \left( \pi^2 - 4 \right).
  \end{align*}

  For $|\, \phi \,| \ge \varphi$, the sign of the derivative
  \begin{align*}
    \frac{\dd}{\dd \phi}
    \left(
      1 - \cos \phi - \frac{2\phi^2}{\pi^2}
    \right)
    =
    \sin \phi - \frac{4\phi}{\pi^2}.
  \end{align*}
  is $\sgn(\phi)$, whilst $1 - \cos \phi > 2\phi^2 / \pi^2 > 0$ when $|\, \phi \,| = \varphi$.

  Hence, $2\phi^2 / \pi^2$ is less than but approaching $1 - \cos \phi$ in $|\, \phi \,|$ when $|\, \phi \,| = \varphi$.
  The crossover point where $2\phi^2 / \pi^2 = 1 - \cos \phi$ occurs when $|\, \phi \,| = \pi$, as is easy to verify.
  This proves the lower bound on $1 - \cos \phi$, and so the claim follows.
\end{proof}

\steponeclaimcostwo*
\begin{proof}
  For $\phi = 0$, the claim trivially holds.
  Given the series expansion
  \begin{align*}
    \cos \phi
    =
    \sum_{k \, = \, 0}^\infty
      \frac{(-1)^k \phi^{2k}}{(2k)!}
    = 1 - \frac{\phi^2}{2!} + \frac{\phi^4}{4!} - \frac{\phi^6}{6!} + \ordo(\phi^8)
  \end{align*}
  we have for $0 < |\, \phi \,| \le \pi$ that
  \begin{align*}
    \frac{\phi^2}{2!} - \frac{\phi^4}{4!} \le 1 - \cos \phi \le \frac{\phi^2}{2!} - \frac{\phi^4}{4!} + \frac{\phi^6}{6!}
  \end{align*}
  as the series is alternating, and as it holds for any $k \ge 1$ that
  \begin{align*}
    \frac{\phi^{2(k+1)}}{(2(k+1))!} < \frac{\phi^{2k}}{(2k)!}.
  \end{align*}

  It follows that
  \begin{align*}
    -\frac{\phi^4}{4!} \le (1 - \cos \phi) - \frac{\phi^2}{2} \le - \frac{\phi^4}{4!} + \frac{\phi^6}{6!}
    \quad
    \Rightarrow
    \quad
    \left|\, (1 - \cos \phi) - \frac{\phi^2}{2} \,\right| \le \frac{\phi^4}{4!}
  \end{align*}
  and so the claim follows.
\end{proof}

\subsection{Proving approximate uniformity}
\label{appendix:proofs-proving-uniformity}

\sumtrigamma*
\begin{proof}
  Divide the sum into three partial sums
  \begin{align*}
    \sum_{t = -B}^{B}
    \frac{1}{(\alpha_0 + rt)^2}
    =
    \sum_{t = -\infty}^{\infty}
    \frac{1}{(\alpha_0 + rt)^2}
    -
    \sum_{t = -\infty}^{-B - 1}
    \frac{1}{(\alpha_0 + rt)^2}
    -
    \sum_{t = B + 1}^{\infty}
    \frac{1}{(\alpha_0 + rt)^2}.
  \end{align*}

  By~\cite[Sect.~6.4.10 on p.~260]{abst}, for $x > 0$, we have that
  \begin{align*}
    \psi^{(n)}(x) = (-1)^{n+1} n! \sum_{t \, = \, 0}^{\infty} (x + t)^{-n-1}
    \quad
    \Rightarrow
    \quad
    \psi'(x) = \psi^{(1)}(x) = \sum_{t \, = \, 0}^{\infty} \frac{1}{(x + t)^2}
  \end{align*}
  for~$\psi^{(n)}$ the~$n$:th derivative of~$\psi$.
  It follows that
  \begin{align*}
    \psi'(1 + B + \alpha_0/r) &= \sum_{t \, = \, 0}^\infty \frac{1}{((1 + B + \alpha_0/r) + t)^2}
                               = \sum_{t \, = \, B + 1}^\infty \frac{1}{(\alpha_0/r + t)^2} \\
                              &= \sum_{t \, = \, B + 1}^\infty \frac{r^2}{r^2 (\alpha_0/r + t)^2}
                               = \sum_{t \, = \, B + 1}^\infty \frac{r^2}{(\alpha_0 + rt)^2}
  \end{align*}
  and analogously that
  \begin{align*}
    \psi'(1 + B - \alpha_0/r) &= \sum_{t \, = \, 0}^\infty \frac{1}{((1 + B - \alpha_0/r) + t)^2}
                               = \sum_{t \, = \, B + 1}^\infty \frac{1}{(-\alpha_0/r + t)^2} \\
                              &= \sum_{t \, = \, B + 1}^\infty \frac{r^2}{r^2 (-\alpha_0/r + t)^2}
                               = \sum_{t \, = \, B + 1}^\infty \frac{r^2}{(-\alpha_0 + rt)^2} \\
                              &= \sum_{t \, = \, -\infty}^{-B - 1} \frac{r^2}{(-\alpha_0 - rt)^2}
                               = \sum_{t \, = \, -\infty}^{-B - 1} \frac{r^2}{(\alpha_0 + rt)^2}.
  \end{align*}

  It now only remains to show that
  \begin{align*}
    \sum_{t = -\infty}^{\infty}
    \frac{1}{(\alpha_0 + rt)^2}
    =
    \frac{1}{r^2} \frac{2\pi^2}{1 - \cos(2 \pi \alpha_0 / r)}.
  \end{align*}
  To this end, we use that by~\cite[Sect.~4.3.92 on p.~75]{abst} it holds that
  \begin{align*}
    \csc^2(x) = \sum_{t = -\infty}^{\infty} \frac{1}{(x + t \pi)^2}
    \quad
    \text{assuming}
    \quad
    x \not\in \{ \pi u \,|\, u \in \mathbb Z\},
  \end{align*}
  where we note when comparing to~\cite{abst} that the sign of the term $t \pi$ in the denominator is arbitrary since we sum over all $t \in \mathbb Z$.
  It follows that
  \begin{align*}
    \pi^2 \csc^2(\pi x) = \sum_{t = -\infty}^{\infty} \frac{1}{(x + t)^2}
    \quad
    \text{assuming}
    \quad
    x \not\in \mathbb Z,
  \end{align*}
  which in turn implies, assuming $x/r \not\in \mathbb Z$, that
  \begin{align*}
    \pi^2 \csc^2(\pi x / r)
    &=
    \sum_{t = -\infty}^{\infty} \frac{1}{(x / r + t)^2}
     =
    \sum_{t = -\infty}^{\infty} \frac{r^2}{r^2 (x / r + t)^2}
     =
    \sum_{t = -\infty}^{\infty} \frac{r^2}{(x + rt)^2}
  \end{align*}
  from which it follows that
  \begin{align*}
    \sum_{t = -\infty}^{\infty} \frac{1}{(\alpha_0 + rt)^2}
    =
    \frac{\pi^2}{r^2} \csc^2(\pi \alpha_0 / r)
    =
    \frac{1}{r^2} \frac{2 \pi^2}{1-\cos(2 \pi \alpha_0 / r)}
  \end{align*}
  where $\alpha_0 / r \notin \mathbb Z$ as $\alpha_0$ is non-zero and on $(-r/2, r/2]$, and so the lemma follows.
\end{proof}

\trigamma*
\begin{proof}
  As in the formulation, let $x \in \mathbb R$ and $x > 0$.
  By~\cite[Sect.~6.4.10 on p.~260]{abst},
  \begin{align*}
    \psi^{(n)}(x) = (-1)^{n+1} n! \sum_{k \, = \, 0}^{\infty} (x + k)^{-n-1}
    \quad
    \Rightarrow
    \quad
    \psi'(x) = \psi^{(1)}(x) = \sum_{k \, = \, 0}^{\infty} \frac{1}{(x + k)^2}
  \end{align*}
  for~$\psi^{(n)}$ the~$n$:th derivative of~$\psi$.
  Let $s \in \mathbb R$ and $s > 1$.
  As Nemes~\cite[Sect.~1]{nemes}~states,
  \begin{align*}
    \zeta(s, x) = \sum_{k \, = \, 0}^{\infty} \frac{1}{(x + k)^s}
    \quad
    \Rightarrow
    \quad
    \psi'(x) = \zeta(2, x),
  \end{align*}
  where $\zeta(s, x)$ is the Hurwitz zeta function.
  We now follow Nemes~\cite{nemes}:

  Let~$N$ be a positive integer.
  By~\cite[eq.~(1.3)]{nemes},
  \begin{align*}
    \zeta(s, x)
    =
    \frac{1}{2} x^{-s} + \frac{x^{1-s}}{s-1} + x^{1-s}
    \left(
      \sum_{n=1}^{N-1} \frac{B_{2n}}{(2n)!} \frac{(s)_{2n-1}}{x^{2n}} + R_N(s, x)
    \right)
  \end{align*}
  for~$B_{2n}$ the $2n$:th Bernoulli number, and $(u)_v = \Gamma(u+v) / \Gamma(u)$ the Pochhammer symbol.
  As for the remainder term $R_N(s, x)$, by~\cite[Thm.~1.3]{nemes},
  \begin{align*}
    R_N(s, x) = \frac{B_{2N}}{(2N)!} \frac{(s)_{2N-1}}{x^{2N}} \theta_N(s, x)
  \end{align*}
  where $\theta_N(s, x) \in (0, 1)$.
  For $s = N = 2$ and any real $x > 0$, it follows that
  \begin{align*}
    \psi'(x)
    =
    \zeta(2, x)
    &=
    \frac{1}{2x^2} + \frac{1}{x} + \frac{1}{x}
    \left(
      \frac{B_{2}}{2} \frac{(2)_{1}}{x^2} + R_2(2, x)
    \right) \\
    &=
    \frac{1}{x} + \frac{1}{2x^2} + \frac{1}{6 x^{3}} + \frac{R_2(2, x)}{x}
    <
    \frac{1}{x} + \frac{1}{2x^2} + \frac{1}{6 x^{3}}
  \end{align*}
  as $(2)_1 = 2$ and $B_2 = \frac{1}{6}$, and where we have used that
  \begin{align*}
    R_2(2, x)
    =
    \frac{B_{4}}{4!} \frac{(2)_{3}}{x^{4}} \theta_2(2, x)
    =
    -\frac{1}{30x^4} \theta_2(2, x) < 0
  \end{align*}
  as $(2)_{3} = 4!$ and $B_4 = -\frac{1}{30}$, and so the claim follows.
\end{proof}

\end{document}